\documentclass{article}
\usepackage{amsfonts,amsmath,amssymb}
\usepackage{dsfont}
\usepackage{graphicx}
\usepackage{multirow}
\usepackage{stmaryrd}
\usepackage{pstricks}
\usepackage[standard]{ntheorem}
\begin{document}

\title{Protein Folding in the 2D Hydrophobic-Hydrophilic (HP) Square 
Lattice Model is Chaotic\thanks{Authors in alphabetical order}
}

\author{Jacques M. Bahi, Nathalie C\^ot\'e, Christophe Guyeux, and Michel Salomon}




\newcommand{\CG}[1]{\begin{color}{red}\textit{#1}\end{color}}

\maketitle

\begin{abstract}

Among the unsolved problems  in computational biology, protein folding
is one  of the  most interesting challenges.   To study  this folding,
tools like neural networks and  genetic algorithms have received a lot
of  attention,  mainly  due  to  the NP-completeness  of  the  folding
process.  The background idea that has  given rise to the use of these
algorithms  is  obviously that  the  folding  process is  predictable.
However, this important assumption is disputable as chaotic properties
of  such a  process have  been recently  highlighted.  In  this paper,
which  is  an  extension  of  a  former  work  accepted  to  the  2011
International  Joint  Conference  on  Neural Networks  (IJCNN11),  the
topological behavior of a well-known dynamical system used for protein
folding  prediction is  evaluated.  It  is  mathematically established
that  the  folding dynamics  in  the  2D hydrophobic-hydrophilic  (HP)
square lattice model, simply called ``the 2D model'' in this document,
is  indeed   a  chaotic  dynamical  system  as   defined  by  Devaney.
Furthermore, the  chaotic behavior of this model  is qualitatively and
quantitatively deepened, by  studying other mathematical properties of
disorder,   namely:   the   indecomposability,   instability,   strong
transitivity,  and  constants of  expansivity  and sensitivity.   Some
consequences  for both biological  paradigms and  structure prediction
using this model are then  discussed.  In particular, it is shown that
some neural  networks seems to be  unable to predict  the evolution of
this model with accuracy, due to its complex behavior.

\end{abstract}

\section{Introduction}
\label{intro}

Proteins, polymers formed  by different kinds of amino  acids, fold to
form a specific tridimensional  shape.  This geometric pattern defines
the  majority of  functionality within  an organism,  \emph{i.e.}, the
macroscopic  properties, function,  and behavior  of a  given protein.
For  instance, the hemoglobin  is able  to carry  oxygen to  the blood
stream  due to  its 3D  geometric pattern.   However, contrary  to the
mapping from DNA  to the amino acids sequence,  the complex folding of
this last  sequence still remains not  well-understood.  Moreover, the
determination  of 3D  protein  structure from  the  amino acid  linear
sequence,  that is  to say,  the  exact computational  search for  the
optimal conformation  of a molecule, is completely  unfeasible.  It is
due  to  the  astronomically  large  number  of  possible  3D  protein
structures  for  a  corresponding  primary  sequence  of  amino  acids
\cite{Hoque09}: the computation  capability required even for handling
a   moderately-sized  folding   transition  exceeds   drastically  the
computational  capacity around  the world.   Additionally,  the forces
involved in  the stability of  the protein conformation  are currently
not modeled  with \linebreak  enough accuracy \cite{Hoque09},  and one
can even wonder if one day a fully accurate model can be found.

Then  it is impossible  to compute  exactly the  3D structures  of the
proteins.  Indeed,  the Protein Structure Prediction  (PSP) problem is
NP-complete \cite{Crescenzi98}.   This is why the  3D conformations of
proteins are \emph{predicted}: the  most stable energy-free states are
looked  for by  using  computational intelligence  tools like  genetic
algorithms      \cite{DBLP:conf/cec/HiggsSHS10},      ant     colonies
\cite{Shmygelska2005Feb},                particle                swarm
\cite{DBLP:conf/cec/Perez-HernandezRG10},      memetic      algorithms
\cite{Islam:2009:NMA:1695134.1695181},      or     neural     networks
\cite{Dubchak1995}.   This  search   is  justified  by  the  Afinsen's
``Thermodynamic  Hypothesis'',   claiming  that  a   protein's  native
structure     is    at    its     lowest    free     energy    minimum
\cite{Anfinsen20071973}.  The use  of computational intelligence tools
coupled  with  proteins   energy  approximation  models  (like  AMBER,
DISCOVER,  or ECEPP/3),  come from  the  fact that  finding the  exact
minimum energy of a 3D structure of a protein is a very time consuming
task.   Furthermore, in  order to  tackle  the complexity  of the  PSP
problem, authors that  try to predict the protein  folding process use
models of various resolutions.  In low resolution models, atoms in the
same amino  acid can  for instance be  considered as the  same entity.
These low  resolution models  are used  as the first  stage of  the 3D
structure  prediction:   the  backbone  of  the   3D  conformation  is
determined.   Then,  high  resolution  models come  next  for  further
exploration.   Such a  prediction  strategy is  commonly  used in  PSP
softwares   like   ROSETTA   \cite{Bonneau01,Chivian2005}  or   TASSER
\cite{Zhang2005}.

In  this   paper,  which  is  an  extension   of  \cite{bgc11:ip},  we
mathematically demonstrate that a particular dynamical system, used in
low  resolutions models  to predict  the backbone  of the  protein, is
chaotic  according to  the  Devaney's formulation.   Chaos in  protein
folding  has  been  already  investigated  in  the  past  years.   For
instance, in  \cite{Bohm1991375}, the  Lyapunov exponent of  a folding
process has been experimentally computed, to show that protein folding
is highly  complex.  More precisely,  the author has  established that
the  crambin protein  folding process,  which  is a  small plant  seed
protein constituted  by 46~amino acids  from \emph{Crambe Abyssinica},
has a  positive Lyapunov exponent.   In \cite{Zhou96}, an  analysis of
molecular dynamics simulation of a model $\alpha$-helix indicates that
the motion  of the helix  system is chaotic, \emph{i.e.},  has nonzero
Lyapunov exponents, broad-band  power spectra, and strange attractors.
Finally,  in  \cite{Braxenthaler97},   the  authors  investigated  the
response of a  protein fragment in an explicit  solvent environment to
very small  perturbations of the  atomic positions, showing  that very
tiny  changes in  initial conditions  are amplified  exponentially and
lead  to vastly different,  inherently unpredictable  behavior.  These
papers have studied experimentally the dynamics of protein folding and
stated  that  this process  exhibits  some  chaotic properties,  where
``chaos'' refers to various physical understandings of the phenomenon.
They noted  the complexity of  the process in concrete  cases, without
offering  a  study framework  making  it  possible  to understand  the
origins of such a behavior.

The approach presented in this  research work is different for the two
following reasons.  First, we  focus on mathematical aspects of chaos,
like the  Devaney's formulation of  a chaotic dynamical  system.  This
well-known topological  notion for  a chaotic behavior  is one  of the
most  established  mathematical  definition  of  unpredictability  for
dynamical  systems.  Second, we  do not  study the  biological folding
process, but the protein folding process  as it is described in the 2D
hydrophobic-hydrophilic (HP) lattice  model \cite{Berger98}.  In other
words,  we mathematically  study  the folding  dynamics  used in  this
model,  and  we   wonder  if  this  model  is   stable  through  small
perturbations.  For instance, what are  the effects in the 2D model of
changing a  residue from hydrophobic to hydrophilic?   Or what happens
if we do not realize exactly the good rotation on the good residue, at
one given stage of the 2D  folding process, due to small errors in the
knowledge of the protein?

Let us recall  that the 2D HP square lattice model  is a popular model
with low resolution that  focuses only on hydrophobicity by separating
the amino  acids into  two sets: hydrophobic  (H) and  hydrophilic (or
polar P) \cite{Dill1985}.  This model  has been used several times for
protein                       folding                       prediction
\cite{DBLP:conf/cec/HiggsSHS10,Braxenthaler97,DBLP:conf/cec/IslamC10,Unger93,DBLP:conf/cec/HorvathC10}.
In  what   follows,  we  show   that  \emph{the  folding   process  is
  unpredictable (chaotic) in  the 2D HP square lattice  model used for
  prediction},  and  we investigate  the  consequences  of this  fact.
Chaos here  refers to our  inability to make relevant  prediction with
this  model,   which  does  not  \emph{necessarily}   imply  that  the
biological folding dynamics is chaotic, too.  In particular, we do not
claim  that  these biological  systems  must  try  a large  number  of
conformations in order  to find the best one.   Indeed, the prediction
model is proven to be chaotic, but this fact is not clearly related to
the  impact  of  environmental  factors  on  true  biological  protein
folding.

\bigskip

After having established by two different proofs the chaos, as defined
in the Devaney's  formulation, of the dynamical system  used in the 2D
model, we will deepen the evaluation of the disorder generated by this
system for backbone prediction.  A qualitative topological study shows
that  its  folding  dynamics  is  both  indecomposable  and  unstable.
Moreover,   the   unpredictability   of   the  system   is   evaluated
quantitatively too,  by computing the  constant of sensibility  to the
initial conditions and the constant of expansivity.
All of these results show  that the dynamical system used for backbone
prediction in the 2D model has  a very intense chaotic behavior and is
highly unpredictable.

Consequences  of these  theoretical results  are then  outlined.  More
precisely,  we will  focus on  the following  questions.   First, some
artificial intelligence tools used  for protein folding prediction are
then based,  for the backbone  evaluation, on a dynamical  system that
presents  several  chaotic properties.   It  is  reasonable to  wonder
whether these  properties impact the quality of  the prediction.  More
specifically, we  will study  if neural networks  are able to  learn a
topological chaotic  behavior, and if predictions  resulting from this
learning   are  close   to   the  reality.    Moreover,  the   initial
conformation,  encompassing   the  sequence  of   amino  acids,  their
interactions, and  the effects of  the outside world, are  never known
with infinite precision.  Taking into  account the fact that the model
used for prediction  embeds a dynamical system being  sensitive to its
initial condition, what can we  conclude about the confidence put into
the final  3D conformation?  Concerning the biological  aspects of the
folding  process, the  following facts  can be  remarked.  On  the one
hand, a  chaotic behavior seems to be  incompatible with approximately
one thousand  general categories  of folds: this  final kind  of order
seems  in  contradiction  with  chaos.  Additionally,  sensibility  to
initial  conditions seems  to be  contradictory with  the fact  that a
sequence  of  amino  acids  always  folds in  the  same  conformation,
whatever the environment  dependency.  So, as the 2D  HP lattice model
for backbone  prediction is chaotic whereas the  whole folding process
seems not, one can wonder  whether this backbone prediction is founded
or not.   On the other  hand, recent experimental  researches recalled
previously tend to  prove that the folding process  presents, at least
to  a  certain extent,  some  characteristics  of  a chaotic  behavior
\cite{Bohm1991375,Zhou96,Braxenthaler97}.  If this theory is confirmed
and  biological  explanations  are  found  (for  instance,  regulatory
processes  could  repair  or  delete misfolded  proteins),  then  this
research work could appear as a first step in the theoretical study of
the chaos of protein folding.

In fact,  the contradiction  raised above is  only apparent, as  it is
wrong to claim that all of the sequences of amino acids always fold in
a  constant and  well-defined  conformation. More  precisely, a  large
number of proteins,  called ``intrinsically unstructured proteins'' or
``intrinsically disordered  proteins'', lay  at least in  part outside
this  rule. More  than 600~proteins  are proven  to be  of  this kind:
antibodies,  p21 and  p27  proteins, fibrinogen,  casein in  mammalian
milk, capsid  of the Tobacco mosaic  virus, proteins of  the capsid of
bacteriophages, to name a few. Indeed, a large number of proteins have
at  least  a  disordered  region  of  greater  or  lesser  size.  This
flexibility allow them to exert  various functions into an organism or
to bind to  various macromolecules. For instance, the  p27 protein can
be binded to various kind  of enzymes.  Furthermore, some studies have
shown that  between 30\%  and 50\% of  the eukaryote proteins  have at
least          one          large         unstructured          region
\cite{Dyson2005,doi:10.1146/annurev.biophys.37.032807.125924}.  Hence,
regular  and disordered  proteins can  be linked  to  the mathematical
notions of chaos  as understood by Devaney, or  Knudsen, which consist
in the  interlocking of points  having a regular behavior  with points
whose desire is to visit the whole space.

The remainder  of this  paper is structured  as follows.  In  the next
section we recall some notations and terminologies on the 2D model and
the  Devaney's  definition of  chaos.   In Section  \ref{sec:dynamical
  system},  the  folding process  in  the 2D  model  is  written as  a
dynamical   system   on   a   relevant   metrical   space.    Compared
to~\cite{bgc11:ip},  we  have  simplified  the  folding  function  and
refined the metrical space to the set of all acceptable conformations.
This work,  which is  the first contribution  of this paper,  has been
realized  by   giving  a  complete  understanding   of  the  so-called
Self-Avoiding Walk  (SAW) requirement.  In Sections~\ref{sec:HP=chaos}
and \ref{sec:CI=chaos}, proofs of  the chaotic behavior of a dynamical
system  used for backbone  prediction, are  taken from~\cite{bgc11:ip}
and adapted to this  set of acceptable conformations.  This adaptation
is the second contribution of  this research work.  The first proof is
directly achieved in  Devaney's context whereas the second  one uses a
previously    proven     result    concerning    chaotic    iterations
\cite{guyeux09}.  The following section  is devoted to qualitative and
quantitative  evaluations of  the  disorder exhibited  by the  folding
process.  This is the third theoretical contribution of this extension
of \cite{bgc11:ip}.   Consequences of this  unpredictable behavior are
given in  Section \ref{Sec:Consequences}.   Among other things,  it is
regarded  whether  chaotic  behaviors   are  harder  to  predict  than
``normal'' behaviors or not, and  if such behaviors are easy to learn.
This   section  extends   greatly  the   premises   outlined  formerly
in~\cite{bgc11:ip}.   Additionally, reasons  explaining why  a chaotic
behavior unexpectedly  leads to approximately  one thousand categories
of folds  are proposed.  This paper  ends by a  conclusion section, in
which  our contribution  is  summarized and  intended  future work  is
presented.

\section{Basic Concepts}
\label{Sec:basic recalls}

In the sequel $S^{n}$ denotes the $n^{th}$ term of a sequence $S$ and
$V_{i}$ the $i^{th}$ component of a vector $V$.
The $k^{th}$
composition of a single function $f$ is represented by $f^{k}=f
\circ...\circ f$.
The set of congruence classes modulo 4 is denoted by $\mathds{Z}/4\mathds{Z}$.
 Finally, given two integers $a<b$, the following notation is used:
$\llbracket a;b\rrbracket =\{a,a+1,\hdots,b\}$.

\subsection{2D Hydrophilic-Hydrophobic (HP) Model}

\subsubsection*{HP Model}

In  the HP model,  hydrophobic interactions  are supposed  to dominate
protein  folding.  This  model was  formerly introduced  by  Dill, who
considers in  \cite{Dill1985} that the protein core  freeing up energy
is formed by hydrophobic  amino acids, whereas hydrophilic amino acids
tend  to move  in the  outer surface  due to  their affinity  with the
solvent (see Fig.~\ref{fig:hpmodel}).
 
In  this model,  a protein  conformation is  a  ``self-avoi\-ding walk
(SAW)'' on a  2D or 3D lattice such that its  energy $E$, depending on
topological neighboring contacts  between hydrophobic amino acids that
are not  contiguous in  the primary structure,  is minimal.   In other
words, for an  amino-acid sequence $P$ of length  $\mathsf{N}$ and for
the set $\mathcal{C}(P)$  of all SAW conformations of  $P$, the chosen
conformation  will   be  $C^*  =   argmin  \left\{E(C)  \big/   C  \in
\mathcal{C}(P)\right\}$ \cite{Shmygelska05}.  In  that context and for
a conformation  $C$, \linebreak  $E(C)=-q$ where $q$  is equal  to the
number of  topological hydrophobic neighbors.   For example, $E(c)=-5$
in Fig.~\ref{fig:hpmodel}.

\begin{figure}[t]
\centering
\includegraphics[width=2.35in]{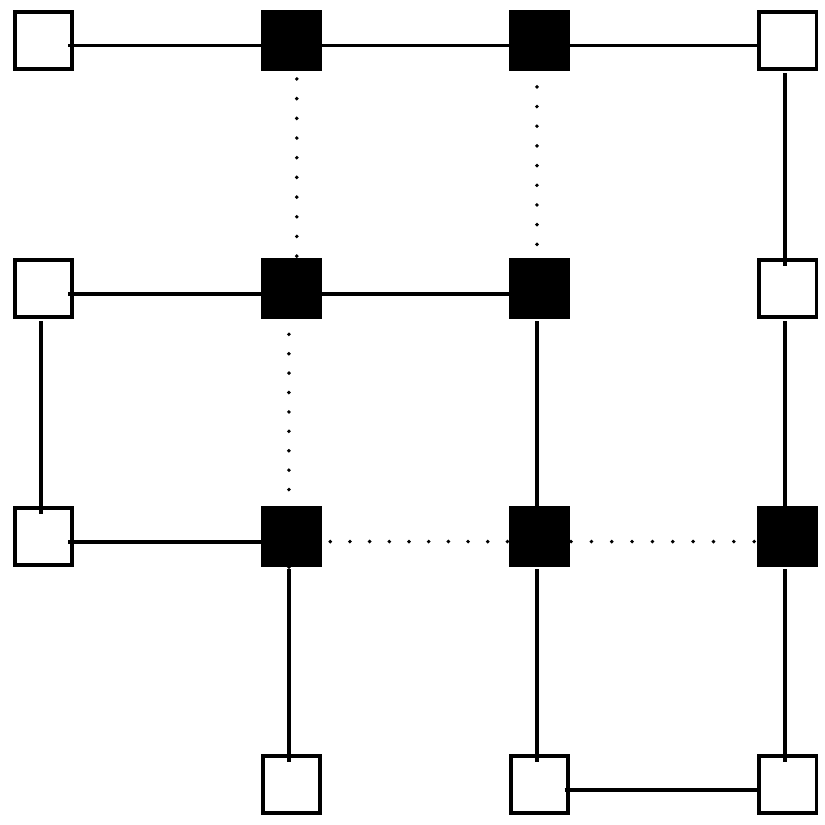}
\caption{Hydrophilic-hydrophobic model (black squares are
hydrophobic residues)}
\label{fig:hpmodel}
\end{figure}

%
%

\subsubsection*{Protein Encoding}

Additionally to the direct coordinate presentation, at least two other
isomorphic encoding  strategies for  HP models are  possible: relative
encoding and absolute  encoding.  In relative encoding \cite{Hoque09},
the  move  direction is  defined  relative  to  the direction  of  the
previous     move.      Alternatively,     in    absolute     encoding
\cite{Backofen99algorithmicapproach}, which is  the encoding chosen in
this paper, the direct  coordinate presentation is replaced by letters
or  numbers  representing  directions  with  respect  to  the  lattice
structure.

For absolute  encoding in the  2D square lattice, the  permitted moves
are:  forward $\rightarrow$  (denoted  by 0),  down $\downarrow$  (1),
backward $\leftarrow$  (2), and up $\uparrow$ (3).   A 2D conformation
$C$ of  $\mathsf{N}+1$ residues for a  protein $P$ is  then an element
$C$ of  $\mathds{Z}/4\mathds{Z}^{\mathsf{N}}$, with a  first component
equal   to    0   (forward)   \cite{Hoque09}.     For   instance,   in
Fig.~\ref{fig:hpmodel},  the 2D  absolute  encoding is  00011123322101
(starting from  the upper  left corner).  In  that situation,  at most
$4^{\mathsf{N}}$   conformations   are   possible   when   considering
$\mathsf{N}+1$ residues, even  if some of them are  invalid due to the
SAW requirement.

\subsection{Devaney's Chaotic Dynamical Systems}
\label{subsection:Devaney}

From  a  mathematical point  of  view,  deterministic  chaos has  been
thoroughly studied  these last decades, with  different research works
that  have   provide  various  definitions  of   chaos.   Among  these
definitions, the  one given  by Devaney~\cite{Devaney} is  perhaps the
most well established.

Consider  a topological  space $(\mathcal{X},\tau)$  and  a continuous
function $f$ on  $\mathcal{X}$.  Topological transitivity occurs when,
for  any point, any  neighborhood of  its future  evolution eventually
overlap with any other given region. More precisely,

\begin{definition}
 $f$ is said to be \emph{topologically transitive} if, for any pair
 of open sets $U,V \subset \mathcal{X}$, there exists $k>0$ such that
 $f^k(U) \cap V \neq \emptyset$.
\end{definition}

This property  implies that a  dynamical system cannot be  broken into
simpler  subsystems.  It  is intrinsically  complicated and  cannot be
simplified.  Besides, a dense set  of periodic points is an element of
regularity that a chaotic dynamical system has to exhibit.

\begin{definition}
 An element (a point) $x$ is a \emph{periodic element} (point) for
 $f$ of period $n\in \mathds{N}^*,$ if $f^{n}(x)=x$.
\end{definition}

\begin{definition}
 $f$ is said to be \emph{regular} on $(\mathcal{X}, \tau)$ if the set
 of periodic points for $f$ is dense in $\mathcal{X}$: for any point
 $x$ in $\mathcal{X}$, any neighborhood of $x$ contains at least one
 periodic point.
\end{definition}

This  regularity ``counteracts'' the  effects of  transitivity.  Thus,
due to these two properties, two points close to each other can behave
in a completely different  manner, leading to unpredictability for the
whole system.  Then,

\begin{definition}[Devaney's chaos]
 $f$ is said to be \emph{chao\-tic} on $(\mathcal{X},\tau)$ if $f$ is
 regular and topologically transitive.
\end{definition}

The chaos property is related to the notion of ``sensitivity'',
defined on a metric space $(\mathcal{X},d)$ by:

\begin{definition} \label{sensitivity} 
 $f$ has \emph{sensitive dependence on initial conditions} if there
 exists $\delta >0$ such that, for any $x\in \mathcal{X}$ and any
 neighborhood $V$ of $x$, there exist $y\in V$ and $n \geq 0$ such
 that $d\left(f^{n}(x), f^{n}(y)\right) >\delta $. 
 
 $\delta$ is called the \emph{constant of sensitivity} of $f$.
\end{definition}

Indeed, Banks  \emph{et al.}  have proven in~\cite{Banks92}  that when
$f$ is chaotic and $(\mathcal{X}, d)$  is a metric space, then $f$ has
the  property  of sensitive  dependence  on  initial conditions  (this
property was formerly an element  of the definition of chaos).  To sum
up, quoting Devaney in~\cite{Devaney}, a chaotic dynamical system ``is
unpredictable   because  of  the   sensitive  dependence   on  initial
conditions.   It  cannot  be   broken  down  or  simplified  into  two
subsystems which do not  interact because of topological transitivity.
And  in the midst  of this  random behavior,  we nevertheless  have an
element  of  regularity''.    Fundamentally  different  behaviors  are
consequently possible and occur in an unpredictable way.

\section{A Dynamical System for the 2D HP Square Lattice Model}
\label{sec:dynamical system}

The objective of  this research work is to  establish that the protein
folding process,  as it is  described in the  2D model, has  a chaotic
behavior.   To  do so,  this  process must  be  first  described as  a
dynamical system.

\subsection{Initial Premises}

Let us start with preliminaries introducing some concepts that will be
useful in our approach.

The  primary structure  of  a given  protein  $P$ with  $\mathsf{N}+1$
residues is coded  by $0 0 \hdots 0$  ($\mathsf{N}$ times) in absolute
encoding.  Its final 2D conformation has an absolute encoding equal to
$0  C_1^*  \hdots C_{\mathsf{N}-1}^*$,  where  $\forall  i, C_i^*  \in
\mathds{Z}/4\mathds{Z}$,  is such  that $E(C^*)  =  argmin \left\{E(C)
\big/ C \in  \mathcal{C}(P)\right\}$.  This final conformation depends
on the repartition  of hydrophilic and hydrophobic amino  acids in the
initial sequence.

Moreover, we suppose that, if  the residue number $n+1$ is forward the
residue number $n$ in absolute  encoding ($\rightarrow$) and if a fold
occurs after  $n$, then the forward  move can only by  changed into up
($\uparrow$) or  down ($\downarrow$).   That means, in  our simplistic
model,  only  rotations of  $+\frac{\pi}{2}$  or $-\frac{\pi}{2}$  are
possible.

Consequently, for a given residue that is supposed to be updated, only
one of  the two possibilities below  can appear for  its absolute move
during a fold:
\begin{itemize}
\item $0 \longmapsto 1, 1 \longmapsto 2, 2 \longmapsto 3,$ or $ 3 \longmapsto 0$ 
for a fold in the clockwise direction, or
\item $1 \longmapsto 0, 2 \longmapsto 1, 3 \longmapsto 2,$ or $0 \longmapsto 3$ 
for an anticlockwise. 
\end{itemize}
This fact leads to the following definition:
\begin{definition}
The   \emph{clockwise    fold   function}   is    the   function   $f:
\mathds{Z}/4\mathds{Z} \longrightarrow \mathds{Z}/4\mathds{Z}$ defined
by $f(x)=x+1 (\textrm{mod}~ 4)$.
\end{definition}
Obviously   the  dual  anticlockwise   fold  function   is  \linebreak
$f^{-1}(x)=x-1 (\textrm{mod}~ 4)$.

Thus at  the $n^{th}$ folding  time, a residue  $k$ is chosen  and its
absolute  move is  changed  by using  either  $f$ or  $f^{-1}$.  As  a
consequence,  all of  the  absolute  moves must  be  updated from  the
coordinate  $k$ until  the last  one  $\mathsf{N}$ by  using the  same
folding function.

\begin{example}
\label{ex1}
If the current conformation is $C=000111$, i.e.,

\begin{figure}[h]
\centering
\includegraphics[width=1.25in]{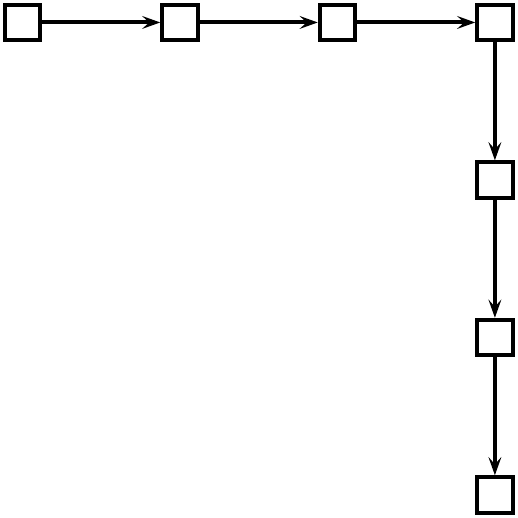}
\end{figure}

\noindent and  if  the  third residue  is  chosen  to  fold  by a  rotation  of
 $-\frac{\pi}{2}$ (mapping $f$), the new conformation will be:
$$(C_1,C_2,f(C_3),f(C_4),f(C_5),f(C_6)) = (0,0,1,2,2,2).$$
\noindent That is,

\begin{figure}[h]
\centering
\includegraphics[width=1.25in]{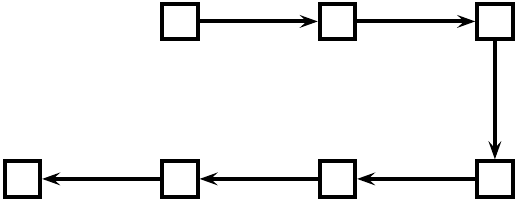}
\end{figure}
\end{example}

These considerations lead to the formalization described hereafter.

\subsection{Formalization and Notations}

Let $\mathsf{N}+1$ be a fixed number of amino acids, where $\mathsf{N}\in\mathds{N}^*$.
We define 
$$\check{\mathcal{X}}=\mathds{Z}/4\mathds{Z}^\mathsf{N}\times \llbracket -\mathsf{N};\mathsf{N} 
\rrbracket^\mathds{N}$$ 
as the phase space of all possible folding processes.
An element $X=(C,F)$ of this dynamical folding space is constituted by:
\begin{itemize}
\item A conformation of the $\mathsf{N}+1$ residues in absolute encoding: $C=(C_1,\hdots, C_\mathsf{N}) \in \mathds{Z}/4\mathds{Z}^\mathsf{N}$. Note that we do not require self-avoiding walks here.
\item A sequence $F \in \llbracket -\mathsf{N} ; \mathsf{N} \rrbracket^\mathds{N}$ of future folds such that, when $F_i \in \llbracket -\mathsf{N}; \mathsf{N} \rrbracket$ is $k$, it means that it occurs:
\begin{itemize}
\item a fold after the $k-$th residue by a rotation of $-\frac{\pi}{2}$ (mapping $f$) at the $i-$th step, if $k = F_i >0$, 
\item no fold at time $i$ if $k=0$,
\item a fold after the $|k|-$th residue by a rotation of $\frac{\pi}{2}$ (\emph{i.e.}, $f^{-1}$) at the $i-$th time, if $k<0$.
\end{itemize}
\end{itemize}
On this phase space, the protein folding dynamic in the 2D model can be formalized as follows.

\medskip

Denote by $i$ the map that transforms a folding sequence in its first term (\emph{i.e.}, in the first folding operation):
$$
\begin{array}{lccl}
i:& \llbracket -\mathsf{N};\mathsf{N} \rrbracket^\mathds{N} & \longrightarrow & \llbracket -\mathsf{N};\mathsf{N} \rrbracket \\
& F & \longmapsto & F^0,
\end{array}$$
by $\sigma$ the shift function over $\llbracket -\mathsf{N};\mathsf{N} \rrbracket^\mathds{N}$, that is to say,
$$
\begin{array}{lccl}
\sigma :& \llbracket -\mathsf{N};\mathsf{N} \rrbracket^\mathds{N} 
 & \longrightarrow & \llbracket -\mathsf{N};\mathsf{N} \rrbracket^\mathds{N} \\
& \left(F^k\right)_{k \in \mathds{N}} & \longmapsto 
 & \left(F^{k+1}\right)_{k \in \mathds{N}},
\end{array}$$
\noindent and by $sign$ the function:
$$
sign(x) = \left\{
\begin{array}{ll}
1 & \textrm{if } x>0,\\
0 & \textrm{if } x=0,\\
-1 & \textrm{else.}
\end{array}
\right.
$$ 
Remark that the shift function removes the first folding operation from the folding sequence $F$ once it has been achieved.

Consider now the map $G:\check{\mathcal{X}} \to \check{\mathcal{X}}$ defined by:
$$G\left((C,F)\right) = \left( f_{i(F)}(C),\sigma(F)\right),$$
\noindent where $\forall k \in \llbracket -\mathsf{N};\mathsf{N} \rrbracket$, 
$f_k: \mathds{Z}/4\mathds{Z}^\mathsf{N} \to \mathds{Z}/4\mathds{Z}^\mathsf{N}$
is defined by: 
\begin{flushleft}
$f_k(C_1,\hdots,C_\mathsf{N}) =$
\end{flushleft}
\begin{flushright}
$ (C_1,\hdots,C_{|k|-1},f^{sign(k)}(C_{|k|}),\hdots,f^{sign(k)}(C_\mathsf{N})).$
\end{flushright}
Thus the folding process of a protein $P$ in the 2D HP square lattice model, with initial conformation equal to $(0,0, \hdots, 0)$ in absolute encoding and a folding sequence equal to $(F^i)_{i \in \mathds{N}}$, is defined by the following dynamical system over $\check{\mathcal{X}}$:
$$
\left\{
\begin{array}{l}
X^0=((0,0,\hdots,0),F)\\
X^{n+1}=G(X^n), \forall n \in \mathds{N}.
\end{array}
\right.
$$

In other  words, at each step  $n$, if $X^n=(C,F)$, we  take the first
folding  operation to  realize, that  is  $i(F) =  F^0 \in  \llbracket
-\mathsf{N};\mathsf{N} \rrbracket$, we update the current conformation
$C$ by rotating all of  the residues coming after the $|i(F)|-$th one,
which means  that we replace the conformation  $C$ with $f_{i(F)}(C)$.
Lastly,  we remove  this  rotation  (the first  term  $F^0$) from  the
folding sequence $F$: $F$ becomes $\sigma(F)$.

\begin{example}
Let  us reconsider Example  \ref{ex1}.  The  unique iteration  of this
folding  process transforms  a point  of $\check{X}$  having  the form
$\left((0,0,0,1,1,1),(+3,     F^1,      F^2,     \hdots)\right)$     in
$G\left(((0,0,0,1,1,1),(+3,F^1,F^2, \hdots))\right),$ which is equal to
$\left((0,0,1,2,2,2),(F^1,F^2, \hdots)\right)$.
\end{example}

\begin{remark}
Such a formalization  allows the study of proteins  that never stop to
fold,  for   instance  due  to  never-ending   interactions  with  the
environment.
\end{remark}

\begin{remark}
A protein $P$ that has finished to fold, if such a protein exists, has
the form $(C,(0,0,0,\hdots))$, where $C$  is the final 2D structure of
$P$.   In this  case, we  can assimilate  a folding  sequence  that is
convergent  to  0,  \emph{i.e.},  of  the form  $(F^0,  \hdots,  F^n,0
\hdots)$, with the finite sequence $(F^0, \hdots, F^n)$.
\end{remark}

We will now introduce the SAW requirement in our formulation of the folding process in the 2D model.

\subsection{The SAW Requirement}

\subsubsection{Towards a Basic SAW Requirement Definition}

Let $\mathcal{P}$ denotes the 2D plane and
$$
\begin{array}{cccc}
p: & \mathds{Z}/4\mathds{Z}^\mathsf{N} & \to & \mathcal{P}^{\mathsf{N}+1} \\
 & (C_1, \hdots, C_\mathsf{N}) & \mapsto & (X_0, \hdots, X_\mathsf{N})
\end{array}
$$
where $X_0 = (0,0)$ and
$$
X_{i+1} = \left\{
\begin{array}{ll}
X_i + (1,0) & ~\textrm{if } c_i = 0,\\
X_i + (0,-1) & ~\textrm{if } c_i = 1,\\
X_i + (-1,0) & ~\textrm{if } c_i = 2,\\
X_i + (0,1) & ~\textrm{if } c_i = 3.
\end{array}
\right.
$$

The map $p$ transforms an  absolute encoding in its 2D representation.
For       instance,       $p((0,0,0,1,1,1))$       is       \linebreak
((0,0);(1,0);(2,0);(3,0);(3,-1);(3,-2);(3,-3)),  that  is,  the  first
figure of Example \ref{ex1}.

Now, for each $(P_0, \hdots, P_\mathsf{N})$ of $\mathcal{P}^{\mathsf{N}+1}$, we denote by $$support((P_0, \hdots, P_\mathsf{N}))$$ the set (with no repetition): $\left\{P_0, \hdots, P_\mathsf{N}\right\}$. For instance,
$$support\left(((0,0);(0,1);(0,0);(0,1))\right) = \left\{(0,0);(0,1)\right\}.$$

Then,

\begin{definition}
\label{def:SAW}
A conformation $(C_1, \hdots, C_\mathsf{N}) \in \mathds{Z}/4\mathds{Z}^{\mathsf{N}}$ satisfies the \emph{self-avoiding walk (SAW) requirement} iff the cardinality of $support(p((C_1, \hdots, C_\mathsf{N})))$ is $\mathsf{N}+1$.
\end{definition}

We can remark that Definition \ref{def:SAW} concerns only one conformation, and not a \emph{sequence} of conformations that occurs in a folding process.

\subsubsection{Understanding the so-called SAW Requirement for a Folding Process}

The next stage in the  formalization of the protein folding process in
the  2D model  as  a dynamical  system  is to  take  into account  the
self-avoiding   walk  (SAW)  requirement,   by  restricting   the  set
$\mathds{Z}/4\mathds{Z}^\mathsf{N}$  of all possible  conformations to
one  of   its  subsets.   That   is,  to  define  precisely   the  set
$\mathcal{C}(P)$ of  acceptable conformations of a  protein $P$ having
$\mathsf{N}+1$ residues.   This stage needs a clear  definition of the
SAW requirement.   However, as stated  above, Definition \ref{def:SAW}
only focus on the SAW requirement  of a given conformation, but not on
a complete folding process.   In our opinion, this requirement applied
to the whole folding process can be understood at least in four ways.

\medskip

In the first  and least restrictive approach, we  call it ``$SAW_1$'',
we  only  require  that  the  studied  conformation  satisfy  the  SAW
requirement of  Definition \ref{def:SAW}.  It is  not regarded whether
this conformation is the result  of a folding process that has started
from  $(0,0,\hdots,0)$.  Such  a SAW  requirement has  been  chosen by
authors   of   \cite{Crescenzi98}    when   they   have   proven   the
NP-completeness of the PSP problem.

The second  approach called $SAW_2$  requires that, starting  from the
initial  condition $(0,0,\hdots,  0)$, we  obtain by  a  succession of
folds a  final conformation  that is a  self-avoiding walk.   In other
words,  we want  that  the final  tree  corresponding to  the true  2D
conformation has  2 vertices with  1 edge and  $\mathsf{N}-2$ vertices
with 2 edges.  For instance,  the folding process of Figure \ref{saw2}
is acceptable in $SAW_2$, even if it presents residues that overlap in
an intermediate conformation. Such an approach corresponds to programs
that start from  the initial conformation $(0,0, \hdots,  0)$, fold it
several times according to  their embedding functions, and then obtain
a final  conformation on which the  SAW property is  checked: only the
last conformation has to satisfy the Definition~\ref{def:SAW}.

\begin{figure}
\centering
\caption{Folding process acceptable in $SAW_2$ but not in $SAW_3$}
\label{saw2}
\includegraphics[width=1.5in]{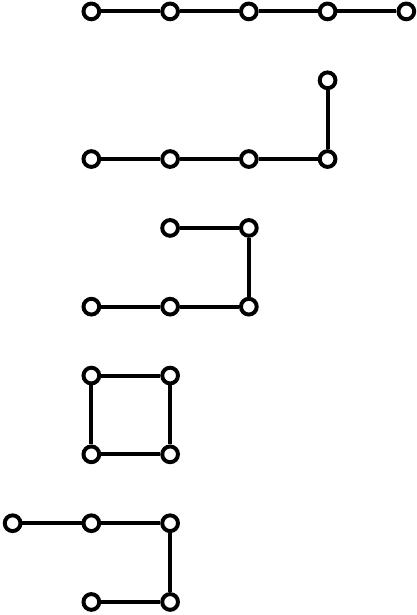}
\end{figure}

In the next  approach, namely the $SAW_3$ requirement,  it is demanded
that each  intermediate conformation, between the initial  one and the
returned  (final)  one,  satisfy  the  Definition  \ref{def:SAW}.   It
restricts        the       set       of        all       conformations
$\mathds{Z}/4\mathds{Z}^\mathsf{N}$, for a  given $\mathsf{N}$, to the
subset   $\mathfrak{C}_\mathsf{N}$   of  conformations   $(C_1,\hdots,
C_\mathsf{N})$ such  that $\exists n \in  \mathds{N}^*,$ $\exists k_1,
\hdots,     k_n     \in     \llbracket     -\mathsf{N};     \mathsf{N}
\rrbracket$, $$(C_1,  \hdots, C_\mathsf{N}) =  G^n\left(((0,0, \hdots,
0),(k_1,  \hdots, k_n))\right)$$ $\forall  i \leqslant n$,
the  conformation  $G^i\left(((0,  \hdots,  0),  (k_1,  \hdots,  k_n))
\right)$ satisfies the Definition \ref{def:SAW}.  This $SAW_3$ folding
process requirement, which is perhaps  the most usual meaning of ``SAW
requirement''  in  the  literature  (it  is  used,  for  instance,  in
\cite{DBLP:conf/cec/HiggsSHS10,Braxenthaler97,DBLP:conf/cec/IslamC10,Unger93,DBLP:conf/cec/HorvathC10}),
has  been  chosen  in  this  research work.   In  this  approach,  the
acceptable  conformations  are  obtained  starting  from  the  initial
conformation $(0,0, \hdots, 0)$ and are such that all the intermediate
conformations satisfy the Definition \ref{def:SAW}.

Finally, the $SAW_4$ approach is  a $SAW_3$ requirement in which there
is no intersection of vertex  or edge during the transformation of one
conformation to  another. For  instance, the transformation  of Figure
\ref{saw4} is  authorized in the  $SAW_3$ approach but refused  in the
$SAW_4$ one:  during the rotation  around the residue identified  by a
cross, the  structure after this residue will  intersect the remainder
of the  ``protein''.  In  this last approach  it is impossible,  for a
protein folding from  one plane conformation to another  plane one, to
use the whole space to achieve this folding.

\begin{figure}
\centering
\caption{Folding process acceptable in $SAW_3$ but not in $SAW_4$}
\label{saw4}
\includegraphics[width=3.25in]{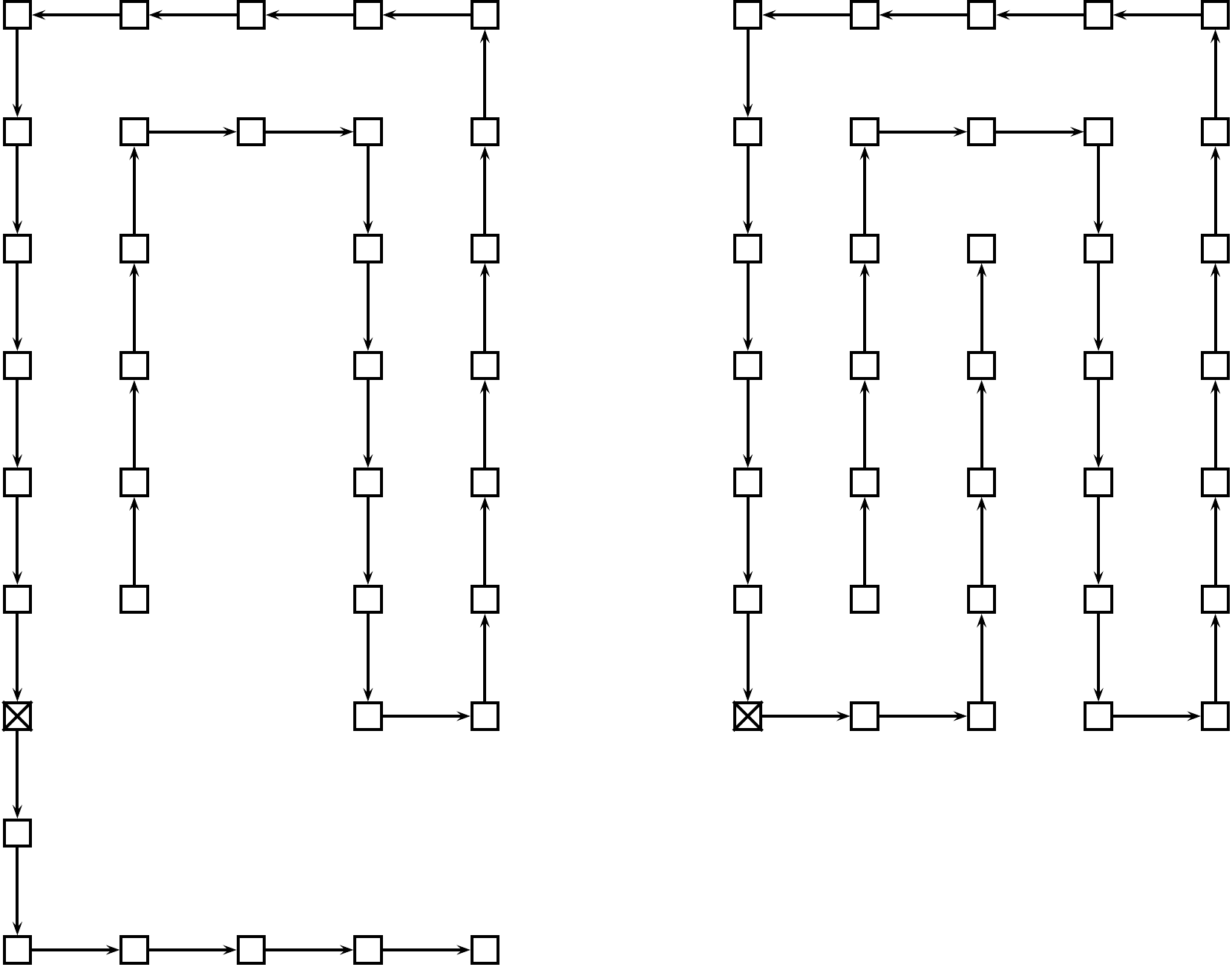}
\end{figure}

Obviously, $SAW_4  \subsetneq SAW_3 \subseteq  SAW_2 \subseteq SAW_1$.
Indeed, it is easy to prove  that $SAW_3 \subsetneq SAW_2$ too, but we
do not  know whether  $SAW_2 \subsetneq SAW_1$  or not.  The  study of
these  four   sets,  their  cardinality,   characterization,  and  the
consequence of  the fact that  the NP-completeness of the  PSP problem
has  been established  in $SAW_1$,  will be  investigated in  a future
work.

In the present  document we cannot decide what  is the most reasonable
approach       between      $SAW_i$,      \linebreak       $i      \in
\left\{1,\hdots,4\right\}$, that is, the  most close to a true natural
protein  folding.   However,  due   to  its  complexity,  the  $SAW_4$
requirement is never  used by tools that embed a  2D HP square lattice
model for protein structure prediction.  That is why we will consider,
in this research work, that the so-called ``SAW requirement'' for a 2D
folding   process  corresponds  to   the  $SAW_3$   approach  detailed
previously.  Indeed,  it is  the most  used one, and  we only  want to
study  the  ability of  PSP  software to  find  the  most probable  2D
conformation.   Thus,   in  what   follows,  the  set   of  acceptable
conformations   with   $\mathsf{N}+1$  residues   will   be  the   set
$\mathfrak{C}_\mathsf{N}$      (\emph{i.e.},     $\mathcal{C}(P)     =
\mathfrak{C}_\mathsf{N}$).

\subsection{A Metric for the Folding Process}

We define a metric $d$ over $\mathcal{X} = \mathfrak{S}_\mathsf{N} \times \llbracket -\mathsf{N};\mathsf{N} \rrbracket^\mathds{N}$ by:
$$\displaystyle{d(X, \check{X}) = d_C(C, \check{C}) + d_F (F, \check{F}).}$$
where
$$ \left\{
\begin{array}{l}
\delta(a,b)=0 \textrm{ if } a=b, \textrm{ otherwise }\delta(a,b)=1, \\
d_C(C, \check{C}) = \displaystyle{\sum_{k=1}^\mathsf{N} \delta(C_k,\check{C}_k) 
 2^{\mathsf{N}-k}}, \\
d_F (F, \check{F}) = \displaystyle{\dfrac{9}{2 \mathsf{N}} \sum_{k=0}^\infty 
 \dfrac{|F^k-\check{F}^k|}{10^{k+1}}.}
\end{array}
\right.$$

This new distance for the dynamical description of the protein folding
process in the 2D HP square lattice model can be justified as follows.
The integral  part of  the distance between  two points  $X=(C,F)$ and
$\check{X}=(\check{C},\check{F})$   of   $\mathcal{X}$  measures   the
differences  between   the  current   2D  conformations  of   $X$  and
$\check{X}$.  More precisely,  if $d_C(C,\check{C})$ is in $\llbracket
2^{N-(k+1)};2^{N-k}  \rrbracket$,  then the  first  $k$  terms in  the
acceptable   conformations  $C$   and   $\check{C}$  (their   absolute
encodings) are  equal, whereas the  $k+1^{th}$ terms differ:  their 2D
conformations will  differ after the \linebreak  $k+1-$th residue.  If
the  decimal part of  $d(X, \check{X})$  is between  $10^{-(k+1)}$ and
$10^{-k}$, then the next k  foldings of $C$ and $\check{C}$ will occur
in the same place (residue),  same order, and same angle.  The decimal
part of $d(X,\check{X})$ will then  decrease as the duration where the
folding process is similar increases.

More  precisely, $F^k  =
\check{F}^k$ (same  residue and same  angle of rotation at  the $k-$th
stage of the  2D folding process) if and only  if the $k+1^{th}$ digit
of this decimal part is  0.  Lastly, $\frac{9}{\mathsf{2N}}$ is just a
normalization factor.

For instance, if we know  where are now the $\mathsf{N}+1$ residues of
our   protein  $P$   in  the   lattice  (knowledge   of   the  correct
conformation), and  if we  have discovered what  will be its  $k$ next
foldings, then we know that the point $X=(C,F)$ describing the folding
process  of  the   considered  protein  in  the  2D   model,  will  be
``somewhere'' into  the ball $\mathcal{B}(C, 10^{-k})$,  that is, very
close to the point $(C,F)$ if $k$ is large.

\begin{example}
Let us consider two points 
\begin{itemize}
\item $X = \left((0,0,0,1,1,1),(3,-4,2)\right)$,
\item and $X' = \left((0,0,0,1,1,1),(3,-4,-6)\right)$ 
\end{itemize}
of $\mathcal{X}$.  We  note $X=(C,F)$ and $X'=(C',F)$.  $d_C(C,C')=0$,
then these two points have  the same current (first) conformation.  As
$d_F(F,F') =  \frac{9}{2\times 6}\frac{|2-(-6)|}{10^3} =  0.006$ is in
$\left[10^{-3};10^{-2}\right[$,  we  can  deduce  that  the  two  next
    foldings of $X$ and of  $X'$ will lead to identical conformations,
    whereas  the  third  folding  operation  will  lead  to  different
    conformations.  A  possible way to  represent these two  points of
    the phase space is to draw the successive conformations induced by
    these points, as  illustrated in Figure \ref{fig:representation du
      phase space}.
\end{example}

\begin{figure}
\centering
\caption{Representation of $X = \left((0,0,0,1,1,1),(3,-4,2)\right)$ and $X' = \left((0,0,0,1,1,1),(3,-4,-6)\right)$ of the phase space $\mathcal{X}$ ($X$ is in left part of the figure, $X'$ is its right part).}
\label{fig:representation du phase space}
\includegraphics[width=2.75in]{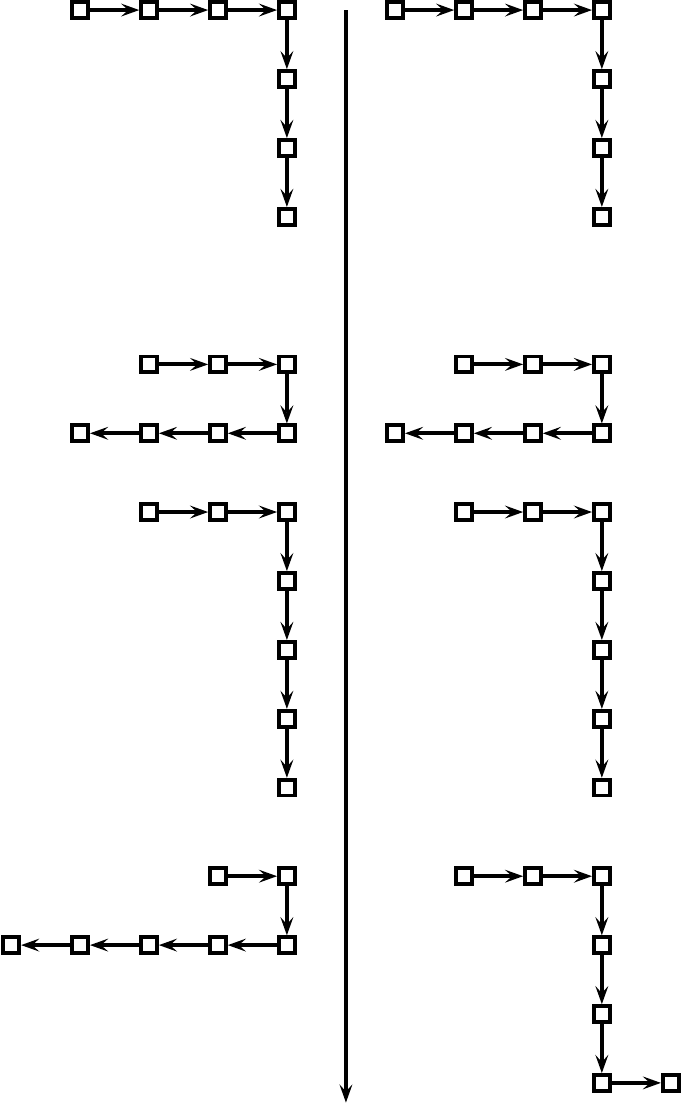}
\end{figure}

\begin{example}
Figure   \ref{fig:representation  du   phase   space2}  contains   the
representation        of       the two       ``points''        $X       =
\left((0,0,0,1,1,1),(3,-4,2)\right)$         and \linebreak        $X'         =
\left((0,0,1,2,2,2),(-4,-5)\right)$.     Let   $X=(C,F)$   and
$X'=(C',F')$.  We have  
$$d_C(C,C') = 2^{6-3}+2^{6-4}+2^{6-5}+2^{6-6} = 15$$
and
$d_F =
\frac{9}{12}\left(\frac{|3-(-4)|}{10}+\frac{|-4-(-5)|}{100}+\frac{|2-0|}{1000}\right)
=    0.534,$
then $d(X,X') =  15.534$.  As 15 is in  $\left[2^3;2^4\right[$, we can
    conclude   that  the   absolute  encodings   of  the   two  initial
    conformations are similar for the first $k=N-4=2$ terms.
\end{example}


\section{Folding Process in 2D Model is Chaotic}
\label{sec:HP=chaos}

\subsection{Motivations}

In our topological  description of the protein folding  process in the
2D model,  all the information  is embedded into the  folding sequence
$F$.   Indeed, roughly speaking,  it is  as if  nature has  a function
$\mathcal{N}$  that   translates  a   protein  $P$  having   a  linear
conformation  $(0,...,0)$  into  an  environment  $E$,  in  a  folding
sequence  $F$, \emph{i.e.},  \linebreak  $F=\mathcal{N}(P,E)$.  Having
this ``natural'' folding sequence~$F$, we  are able to obtain its true
conformation      in     the      2D      model,     by      computing
$G^n\left(((0,\hdots,0),F)\right)$, where $n$ is  the size of $F$.  On
our side, we have only a  partial knowledge of the environment $E$ and
of  the  protein $P$  (exact  interactions  between  atoms).  We  thus
consider $\check{P}$  and $\check{E}$, as close  as we can  to $P$ and
$E$    respectively.     Moreover,    we    have    only    a    model
$\check{\mathcal{N}}$  of  $\mathcal{N}$  as,  for  instance,  we  use
various  approximations:  models for  free  energy, approximations  of
hydrophobic/hydrophilic areas and  electro-polarity, etc.  This is why
we        can        only        deduce        an        approximation
$\check{F}=\check{\mathcal{N}}(\check{P},\check{E})$  of  the  natural
folding  sequence $F=\mathcal{N}(P,E)$.   One important  motivation of
this work is to determine whether, having an approximation $\check{F}$
of   $F$,    we   obtain    a   final   conformation    $\check{C}   =
G^{\check{n}}\left(((0,\hdots,0),\check{F})\right)_0$   close  to  the
natural  conformation  $C  = G^{n}\left(((0,\hdots,0),F)\right)_0$  or
not.  In this last sentence, $n$  and $\check{n}$ are the sizes of $F$
and  $\check{F}$  respectively, and  the  terms ``approximation''  and
``close'' can be understood by using $d_F$ and $d_C$, respectively. To
sum up,  even if we cannot  have access with an  infinite precision to
all  of   the  forces  that   participate  to  the   folding  process,
\emph{i.e.},  even  if  we   only  know  an  approximation  ${X'}^0  =
\left((0,\hdots,0),\check{F}\right)$              of             $X^0=
\left((0,\hdots,0),F\right)$,   can  we   claim  that   the  predicted
conformation                       ${X'}^{n_1}                       =
G^{n_1}\left(((0,\hdots,0),\check{F})\right)$  still remains  close to
the          true          conformation          ${X}^{n_2}          =
G^{n_2}\left(((0,\hdots,0),F)\right)$?   Or, on  the  contrary, do  we
have a chaotic behavior, a kind of butterfly effect that magnifies any
error on the evaluation of the forces in presence?

Raising such a  question leads to the study  of the dynamical behavior
of the folding process.

\begin{figure}
\centering
\caption{Representation  of $X  = \left((0,0,0,1,1,1),(3,-4,2)\right)$
  and  $X' =  \left((0,0,1,2,2,2),(-4,-5)\right)$ of  the  phase space
  $\mathcal{X}$ ($X$ is in left part  of the figure, $X'$ is its right
  part).}
\label{fig:representation du phase space2}
\includegraphics[width=3in]{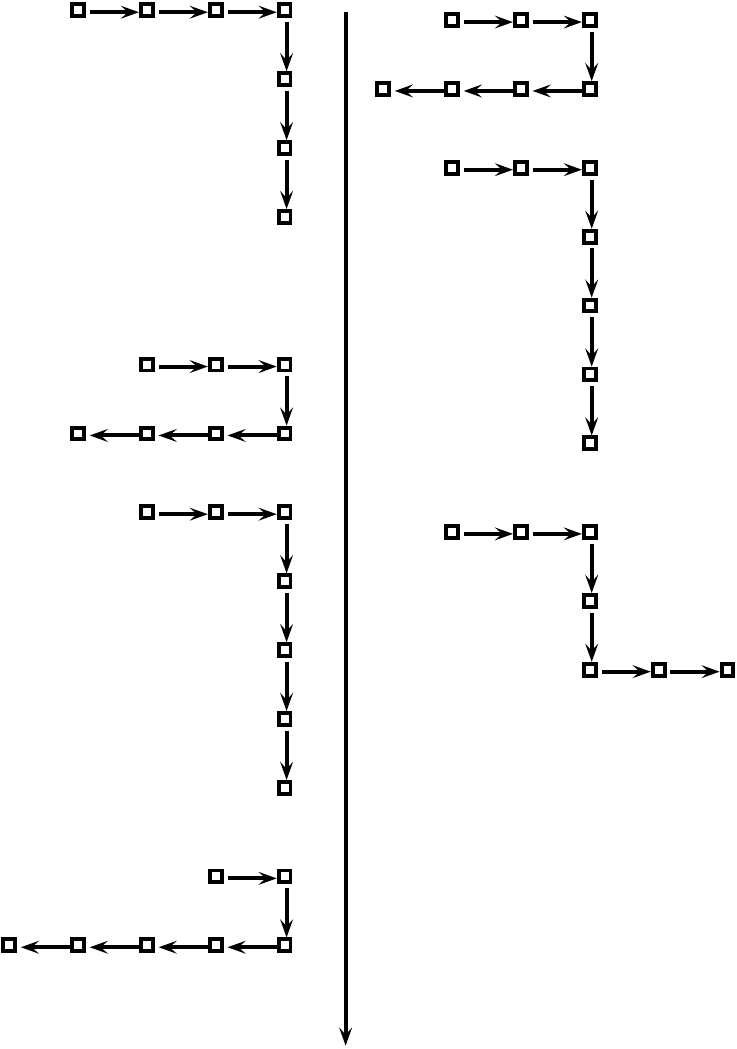}
\end{figure}

\subsection{Continuity of the Folding Process}

We will now give a first  proof of the chaotic behavior of the protein
folding dynamics in  the 2D model.  To do so,  we must establish first
that  $G$  is  a  continuous  map on  $(\mathcal{X},d)$.  Indeed,  the
mathematical theory of chaos only studies dynamical systems defined by
a  recurrence   relation  of  the  form   $X^{n+1}=G(X^n)$,  with  $G$
continuous.

\begin{proposition}
$G$ is a continuous map on $(\mathcal{X},d)$.
\end{proposition}

\begin{proof}
We  will use the  sequential characterization  of the  continuity. Let
$(X^n)_{n \in \mathds{N}}  = \left((C^n,F^n)\right)_{n \in \mathds{N}}
\in   \mathcal{X}^\mathds{N},$  such   that  $X^n   \rightarrow   X  =
(\check{C},\check{F})$.  We  will  then show  that  $G\left(X^n\right)
\rightarrow G(X)$. Let us remark  that $\forall n \in \mathds{N}, F^n$
is a sequence: $F$ is thus a sequence of sequences.

On the one hand, as $X^n=(C^n,F^n) \rightarrow (\check{C},\check{F})$,
we have $d_C\left(C^n,\check{C}\right) \rightarrow 0$, thus $\exists
n_0 \in \mathds{N},$ $n \geqslant n_0$ $\Rightarrow
d_C(C^n,\check{C})=0$. That is, $\forall n \geqslant n_0$ and $\forall
k \in \llbracket 1;\mathsf{N} \rrbracket$, $\delta(C_k^n,\check{C}_k)
= 0$, and so $C^n = \check{C}, \forall n \geqslant n_0.$ Additionally,
since $d_F(F^n,\check{F}) \rightarrow 0$, $\exists n_1 \in \mathds{N}$
such that we have $d_F(F^n_1, \check{F}) \leqslant \frac{1}{10}$. As
a consequence, $\exists n_1 \in \mathds{N},$ $\forall n \geqslant
n_1$, the first term of the sequence $F^n$ is $\check{F}^0$: $i(F^n) =
i(\check{F})$. So, $\forall n \geqslant max(n_0,n_1),$
$f_{i(F^n)}\left(C^n\right)=
f_{i\left(\check{F}\right)}\left(\check{C}\right)$, and then
$f_{i(F^n)}\left(C^n\right)$ $\rightarrow$
$f_{i\left(\check{F}\right)}\left(\check{C}\right)$.

On  the  other   hand,  $\sigma(F^n)  \rightarrow  \sigma(\check{F})$.
Indeed,  $F^n   \rightarrow  \check{F}$  implies  $\sum_{k=0}^{\infty}
\frac{|  \left(F^n\right)^k-\check{F}^k  |}{10^{k+1}} \rightarrow  0$,
from  which   we  obtain  $\frac{1}{10}   \sum_{k=0}^{\infty}  \frac{|
  \left(F^n\right)^{k+1}-\check{F}^{k+1}  |}{10^{k+1}} \rightarrow 0$,
so   $\sum_{k=0}^{\infty}   \frac{|  \sigma(F^n)^k-\sigma(\check{F})^k
  |}{10^{k+1}}$ converges towards 0. Finally, $\sigma(F^n) \rightarrow
\sigma(\check{F})$.

Since we have shown that $f_{i(F^n)}\left(C^n\right)$ $\rightarrow$
$f_{i\left(\check{F}\right)}\left(\check{C}\right)$ and $\sigma(F^n)
\rightarrow \sigma(\check{F})$, we conclude that
$G\left(X^n\right) \rightarrow G(X)$.
\end{proof}

It is now possible to study the chaotic behavior of the folding
process.

\subsection{A First Fundamental Lemma}

Let us  start by introducing the following  fundamental lemma, meaning
that we can transform any  acceptable conformation to any other one in
$SAW_3$, by finding a relevant folding sequence.

\begin{lemma}
\label{lem}
$\forall  C,C'$  in   $\mathfrak{C}_\mathsf{N},$  $\exists  n  \in
  \mathds{N}^*$  and $k_1, \hdots,  k_n$ in  $\llbracket -\mathsf{N};
\mathsf{N}    \rrbracket$    s.t.   
$$G^n\left((C,(k_1,    \hdots,
k_n,0,\hdots))\right) = \left(C',(0,\hdots,0)\right).$$
\end{lemma}

\begin{proof}
As we consider conformations of $\mathfrak{C}_\mathsf{N}$, we take place in the $SAW_3$ requirement, and thus there exist \linebreak $n_1, n_2 \in \mathds{N}^*$ and $l_1, \hdots, l_{n_1}, m_1, \hdots, m_{n_2}$ in $\llbracket -\mathsf{N}; \mathsf{N} \rrbracket$ such that $C = G^{n_1}\left(((0,...,0),(l_1, \hdots, l_{n_1}))\right)$ and \linebreak $C' = G^{n_2}\left(((0,...,0),(m_1, \hdots, m_{n_2}))\right)$.
The result of the lemma is then obtained with $$(k_1, \hdots, k_n) = (-l_{n_1},-l_{n_1-1},\hdots,-l_1,m_1,\hdots, m_{n_2}).$$
\end{proof}

\subsection{Regularity and Transitivity}

Let  us recall  that the  first component  $X_0$ of  $X=(C,F)$  is the
current conformation $C$ of the protein and the second component $X_1$
is its future folding process $F$. We will now prove that,

\begin{proposition}
Folding process in 2D model is regular.
\end{proposition}

\begin{proof}
Let $X=(C,F) \in \mathcal{X}$ and $\varepsilon > 0$. 
Then we define $k_0=-\lfloor log_{10} (\varepsilon) \rfloor$ and $\tilde{X}$ such that:
\begin{enumerate}
\item $\tilde{X}_0 = C$,
\item $\forall k \leqslant k_0, G^k(\tilde{X})_1 = G^k(X)_1$,
\item $\forall i \in \llbracket 1; n \rrbracket,
G^{k_0+i}(\tilde{X})_1 = k_i$,
\item $\forall i \in \mathds{N}, G^{k_0+n+i+1}(\tilde{X})_1 =
G^i(\tilde{X})_1$,
\end{enumerate}
where $k_1, \hdots, k_n$ are integers given by Lemma \ref{lem} with $C=G^{k_0}(X)_0$ and $C'= X_0$. 
Such an $\tilde{X}$ is a periodic point for $G$ into the ball $\mathcal{B}(X,\varepsilon)$. 
(1) and (2) are to make $\tilde{X}$
$\varepsilon-$close to $X$, (3) is for mapping the conformation $G^{k_0}(\tilde{X})_0$ into $C$ in at most $n$ foldings. 
Lastly, (4) is for the periodicity of the folding process.
\end{proof}


Let  us now  consider the  second property  required in  the Devaney's
definition.   Instead of  proving  the transitivity  of  $G$, we  will
establish its strong transitivity:

\begin{definition}
A dynamical system $\left( \mathcal{X}, f\right)$ is strongly
transitive if $\forall x,y \in \mathcal{X},$ $\forall r > 0,$ $\exists
z \in \mathcal{X},$ $d(z,x) \leqslant r \Rightarrow$ $\exists n \in
\mathds{N}^*,$ $f^n(z)=y$.
\end{definition}

In other words, for all $x,y \in \mathcal{X}$, it is possible to find
a point $z$ in the neighborhood of $x$ such that an iterate $f^n(z)$
is $y$.
Obviously, strong transitivity implies transitivity. Let us now prove
that,

\begin{proposition}
Folding process in 2D model is strongly transitive.
\end{proposition}

\begin{proof}
Let $X_A=(C_A,F_A)$, $X_B=(C_B, F_B)$,  and $\varepsilon > 0$. We will
show that $X \in  \mathcal{B}\left(X_A, \varepsilon\right)$ and $n \in
\mathds{N}$ can be found such  that $G^n(X)=X_B$. Let $k_0 = - \lfloor
log_{10}  (\varepsilon  )  \rfloor$  and  $\check{X}=G^{k_0}(C_A,F_A)$
denoted  by  $\check{X}=(\check{C},\check{F})$.   According  to  Lemma
\ref{lem} applied to $\check{C}$ and $C_B$, $\exists k_1, \hdots, k_n$
in    $\llbracket    -\mathsf{N},    \mathsf{N}    \rrbracket$    such
that  $$G^n\left((\check{C},  (k_1,  \hdots,  k_n,0,\hdots))\right)  =
\left(C_B, (0, \hdots )\right).$$

Let us define $X=(C,F)$ in the following way: 
\begin{enumerate}
\item $C=C_A$,
\item $\forall k \leqslant k_0, F^k=F_A^k$,
\item $\forall i \in \llbracket 1; n \rrbracket, F^{k_0+i} =
k_i$,
\item $\forall i \in \mathds{N}, F^{k_0+n+i+1}=F_B^i$.
\end{enumerate}
This point $X$ is thus an element of $\mathcal{B}(X_A,\varepsilon)$
(due to $1,2$) being such that $G^{k_0+n+1}(X) = X_B$ (by
using $3,4$). As a consequence, $G$ is strongly transitive.
\end{proof}

Strong transitivity states that being as close as possible of the true
folding process (2D model) is not a guarantee of success.  Indeed, let
$P$ be a protein under interest and $F$ its natural folding process in
the 2D model.   Then, for all possible conformation  $C$ of the square
lattice, there exists a folding sequence $\check{F}$ very close to $F$
leading to $C$.   More precisely, for any $\varepsilon  > 0$ (as small
as  possible),  an  infinite   number  of  folding  sequences  are  in
$\mathcal{B}_{d_F}(F,\varepsilon)$  and  lead   to  $C$.   The  strong
transitivity  property  implies  that  without the  knowledge  of  the
\emph{exact} initial condition (the  natural folding process, and thus
the  exact   free  energy),   all  the  conformations   are  possible.
Additionally, no  conformation of the square lattice  can be discarded
when studying a protein folding in the 2D HP square lattice model: the
dynamical  system obtained  by such  a formalization  is intrinsically
complicated and cannot be decomposed or simplified.  Furthermore, this
trend  to  visit  the  whole  space  of  acceptable  conformations  is
counteracted  by elements  of  regularity stated  before:  it is  even
impossible to dress a kind  of qualitative description of the dynamics
in the 2D square lattice model,  as two points close to each other can
have fundamentally different behaviors.

\subsection{Chaotic behavior of the folding process}

As $G$ is regular and (strongly) transitive, we have:

\begin{theorem}
\label{leth}
The folding process $G$ in the 2D model is chaotic according to Devaney.
\end{theorem}

Consequently  this   process  is  highly  sensitive   to  its  initial
conditions.   If the  2D  model can  accurately  describe the  natural
process, then Theorem~\ref{leth} implies that even a minute difference
on an intermediate conformation of  the protein, in forces that act in
the  folding process,  or in  the  position of  an atom,  can lead  to
enormous differences in its final conformation, even over fairly small
timescales.  This  is the so-called butterfly  effect.  In particular,
it seems very difficult to predict, in this 2D model, the structure of
a given  protein by  using the knowledge  of the structure  of similar
proteins.  Let us  remark that the whole 3D  folding process with real
torsion angles is  obviously more complex than this  2D HP model.  And
finally, that chaos refers to  our incapacity to make good prediction,
it does not mean that the biological process is a random one.

Before  studying some  practical aspects  of this  unpredictability in
Section \ref{Sec:Consequences}, we will initiate a second proof of the
chaotic behavior of this process and deepen its chaotic properties.

\section{Outlines of a second proof}
\label{sec:CI=chaos}


In this section a second proof of the chaotic behavior of the protein
folding process is given. It is proven that the folding dynamics can
be modeled as chaotic iterations (CIs). CIs are a tool used in
distributed computing and in the computer science security field
\cite{guyeuxTaiwan10} that has been established to be chaotic
according to Devaney \cite{guyeux10}.

This second proof is the occasion to introduce these CIs, which will
be used at the end of this paper to study whether a chaotic behavior
is really more difficult to learn with a neural network than a
``normal'' behavior.

\subsection{Chaotic Iterations: Recalls of Basis}

Let us consider a \emph{system} with a finite number $\mathsf{N} \in
\mathds{N}^*$ of elements (or \emph{cells}), so that each cell has a
Boolean \emph{state}. A sequence of length $\mathsf{N}$ of Boolean
states of the cells corresponds to a particular \emph{state of the
system}. A sequence, which elements are subsets of $\llbracket
1;\mathsf{N} \rrbracket $, is called a \emph{strategy}. The set of all
strategies is denoted by $\mathbb{S}$ and the set $\mathds{B}$ is for
the Booleans $\{0,1\}$.

\begin{definition}
\label{Def:chaotic iterations}
Let $f:\mathds{B}^{\mathsf{N}}\longrightarrow \mathds{B}^{\mathsf{N}}$
be a function and $S\in \mathbb{S}$ be a strategy. The so-called
\emph{chaotic iterations} (CIs) are defined by $x^0\in
\mathds{B}^{\mathsf{N}}$ and $\forall n\in \mathds{N}^{\ast }, \forall i\in
\llbracket1;\mathsf{N}\rrbracket ,$
$$ 
x_i^n=\left\{
\begin{array}{ll}
 x_i^{n-1} & \text{ if } i \notin S^n \\ 
 \left(f(x^{n-1})\right)_i & \text{ if } i \in S^n.
\end{array}\right.
$$
\end{definition}

In other words, at the $n^{th}$ iteration, only the $S^{n}-$th cells
are ``iterated''. Let us notice that the term ``chao\-tic'', in the
name of these iterations, has \emph{a priori} no link with the
mathematical theory of chaos recalled previously. We will now recall
that CIs can be written as a dynamical system, and characterize
functions $f$ such that their CIs are chaotic according to Devaney
\cite{guyeux09}.

\subsection{CIs and Devaney's chaos}

Let
$f: \mathds{B}^{\mathsf{N}}\longrightarrow \mathds{B}^{\mathsf{N}}$.
We define $F_{f}:$
$\llbracket1;\mathsf{N}\rrbracket\times \mathds{B}^{\mathsf{N}}\longrightarrow
\mathds{B}^{\mathsf{N}}$
by:
$$
F_{f}(k,E)=\left( E_{j} \cdot\delta (k,j)+f(E)_{k} \cdot
\overline{\delta (k,j)}\right)_{j\in \llbracket1;\mathsf{N}\rrbracket},
$$
where + and $\cdot$ are the Boolean addition and product operations,
and $\overline{x}$ is for the negation of $x$.

We have proven in \cite{guyeux09} that chaotic iterations can be
described by the following dynamical system:
$$
\left\{
\begin{array}{l}
X^{0}\in \tilde{\mathcal{X}} \\
X^{k+1}=\tilde{G}_{f}(X^{k}),
\end{array}
\right.
$$
where $\tilde{G}_{f}\left((S,E)\right) =\left( \sigma
(S),F_{f}(i(S),E)\right)$, and $\tilde{\mathcal{X}}$ is a metric space
for an ad hoc distance such that $\tilde{G}$ is continuous on
$\mathcal{X}$ \cite{guyeux09}.

Let us now consider the following oriented graph, called \emph{graph
of iterations}. Its vertices are the elements of
$\mathds{B}^\mathsf{N}$, and there is an arc from $x = (x_1, \hdots,
x_i, \hdots, x_\mathsf{N}) \in \mathds{B}^\mathsf{N}$ to $x = (x_1,
\hdots, \overline{x_i}, \hdots, x_\mathsf{N})$ if and only if
$F_f(i,x) = (x_1, \hdots, \overline{x_i}, \hdots, x_\mathsf{N})$. If
so, the label of the arc is $i$. In the following, this graph of
iterations will be denoted by $\Gamma(f)$.

We have proven in \cite{bcgr11:ip} that:

\begin{theorem}
\label{Th:Caracterisation des IC chaotiques}
Functions $f : \mathds{B}^{n} \to \mathds{B}^{n}$ such that
$\tilde{G}_f$ is chaotic according to Devaney, are functions such that
the graph $\Gamma(f)$ is strongly connected.
\end{theorem}

We will now show that the protein folding process can be modeled as
chaotic iterations, and conclude the proof by using the theorem
recalled above. 

\subsection{Protein Folding as Chaotic Iterations}

The  attempt to  use  chaotic  iterations in  order  to model  protein
folding  can be  justified as  follows.  At each  iteration, the  same
process is  applied to the system (\emph{i.e.},  to the conformation),
that is  the folding  operation. Additionally, it  is not  a necessity
that all of the residues fold at each iteration: indeed it is possible
that, at  a given iteration, only  some of these  residues folds. Such
iterations, where not all the cells of the considered system are to be
updated, are exactly the iterations modeled by CIs.

Indeed, the protein folding process with folding sequence $(F^n)_{n
\in \mathds{N}}$ consists in the following chaotic iterations: $C^0 =
(0,0, \hdots, 0)$ and,
$$
C_{|i|}^{n+1} = \left\{
\begin{array}{ll}
C_{|i|}^n & \textrm{if } i \notin S^n\\
f^{sign(i)}(C^n)_i & \textrm{else}
\end{array}
\right.,
$$ 
where the chaotic strategy is defined by $\forall n \in \mathds{N},
\linebreak S^n=\llbracket -\mathsf{N}; \mathsf{N} \rrbracket \setminus
\llbracket -F^n; F^n \rrbracket$.

Thus, to prove that the  protein folding process is chaotic as defined
by Devaney, is equivalent to prove that the graph of iterations of the
CIs defined above is strongly connected. This last fact is obvious, as
it  is  always  possible  to  find  a folding  process  that  map  any
conformation $(C_1, \hdots, C_\mathsf{N}) \in \mathfrak{C}_\mathsf{N}$
to   any   other  \linebreak   $(C_1',   \hdots,  C_\mathsf{N}')   \in
\mathfrak{C}_\mathsf{N}$ (this is Lemma \ref{lem}).

Let us finally remark that it is easy to study processes s.t. more
than one fold occur per time unit, by using CIs. This point will be
deepened in a future work. We will now investigate some consequences
resulting from the chaotic behavior of the folding process.

\section{Qualitative and quantitative evaluations}

Behaviors qualified as ``chaos'' are too complicated to be encompassed
by only one  rigorous definition, as perfect as  it could be.  Indeed,
the  mathematical   theory  of  chaos   brings  several  nonequivalent
definitions  for a  complex, unpredictable  dynamical system,  each of
them  highlighting this  complexity in  a well-defined  but restricted
understanding.   This  is  why,  in  this  section,  we  continue  the
evaluation of the chaotic behavior  of the 2D folding dynamical system
initiated by the proof of the Devaney's chaos.

\subsection{Qualitative study}
\label{QUALITATIVE MEASURE}

First of all, the transitivity property implies the indecomposability
of the system:

\begin{definition}
A dynamical system $\left( \mathcal{X}, f\right)$ is indecomposable if
it is not the union of two closed sets $A, B \subset \mathcal{X}$ such
that $f(A) \subset A, f(B) \subset B$.
\end{definition}

Thus it is  impossible to reduce, in the 2D model,  the set of protein
foldings  in  order  to  simplify its  complexity.   Furthermore,  the
folding process has the instability property:

\begin{definition}
A dynamical system $\left( \mathcal{X}, f\right)$ is unstable if for
all $x \in \mathcal{X}$, the orbit $\gamma_x:n \in \mathds{N}
\longmapsto f^n(x)$ is unstable, that is: $\exists \varepsilon > 0,$
$\forall \delta > 0,$ $\exists y \in \mathcal{X},$ $\exists n \in
\mathds{N},$ $d(x,y) < \delta$ and $d\left(\gamma_x(n),
\gamma_y(n)\right) \geqslant \varepsilon.$
\end{definition}

This property,  which is  implied by the  sensitive dependence  to the
initial conditions, leads to the fact that in all of the neighborhoods
of  any  $x$, there  are  points that  are  separated  from $x$  under
iterations of $f$. We thus can  claim that the behavior of the folding
process is unstable.

\subsection{Quantitative measures}
\label{QUANTITATIVE MEASURE}

\label{par:Sensitivity}
One of the most famous measures in the theory of chaos is the constant
of sensitivity  given in Definition  \ref{sensitivity}. Intuitively, a
function  $f$ having  a constant  of  sensitivity equal  to $\delta  $
implies that  there exists points  arbitrarily close to any  point $x$
that \emph{eventually} separate  from $x$ by at least  $\delta $ under
some iterations of  $f$. This induces that an  arbitrarily small error
on  an initial condition  \emph{may} be  magnified upon  iterations of
$f$.  The  sensitive  dependence   on  the  initial  conditions  is  a
consequence   of   regularity   and   transitivity   in   a   metrical
space~\cite{Banks92}.  However, the  constant of  sensitivity $\delta$
can be obtained by proving the property without using Banks' theorem.

\begin{proposition}
Folding process  in the 2D  model has sensitive dependence  on initial
conditions on $(\mathcal{X},d)$ and  its constant of sensitivity is at
least equal to $2^{\mathsf{N}-1}$.
\end{proposition}

\begin{proof}
Let $X = (C,F) \in \mathcal{X}$, $r>0$, $B =
\mathcal{B}\left(X,r\right)$ an open ball centered in $X$, and $k_0
\in \mathds{Z}$ such that $10^{-k_0-1} \leqslant r < 10^{-k_0}$. We
define $\tilde{X}$ by:
\begin{itemize}
\item $\tilde{C} = C$,
\item $\tilde{F}^k = F^k$, $\forall k \in \mathds{N}$ such that $k
\leqslant k_0$,
\item $\tilde{F}^{k_0+1} = 1$ if $\left|F^{k_0+1}\right| \neq 1$, and
$\tilde{F}^{k_0+1} = - F^{k_0+1}$ else.
\item $\forall k \geqslant k_0+2, \tilde{F}^{k} = - F^{k}$.
\end{itemize} 
Only two cases can occur:
\begin{enumerate}
\item If $\left|F^{k_0+1}\right| \neq 1$, then
\begin{flushleft}
$d\left( G^{k_0+1}\left(X\right), G^{k_0+1}\left(\tilde{X}\right)\right)$
\end{flushleft}
\begin{flushright}
$\begin{array}{l}
 \displaystyle{= 2^{\mathsf{N}-1}+2^{\mathsf{N}-F^{k_0+1}}+\dfrac{9}{2\mathsf{N}} 
 \sum_{k=k_0+1}^{\infty} \dfrac{\left| F^k - \tilde{F}^k \right|}{10^{k+1}}}\\
 \displaystyle{= 2^{\mathsf{N}-1}+2^{\mathsf{N}-F^{k_0+1}}+\dfrac{9}{2\mathsf{N}} 
 \sum_{k=k_0+1}^{\infty} \dfrac{2 \mathsf{N}}{10^{k+1}}}\\ 
\displaystyle{= 2^{\mathsf{N}-1}+2^{\mathsf{N}-F^{k_0+1}}+9 \dfrac{1}{10^{k_0+2}} 
 \dfrac{1}{1 - \dfrac{1}{10}} }\\
\displaystyle{= 2^{\mathsf{N}-1}+2^{\mathsf{N}-F^{k_0+1}}+\dfrac{1}{10^{k_0+1}}.}\\ 
\end{array}$
\end{flushright}
\item Else, $d\left( G^{k_0+1}\left(X\right), G^{k_0+1}\left(\tilde{X}\right)\right) 
 = 2^{\mathsf{N}-1}+\dfrac{1}{10^{k_0+1}}.$
\end{enumerate}

In all of these cases, the sensibility to the initial condition is
greater than $2^{\mathsf{N}-1}$.
\end{proof}

Let us now recall another common quantitative measure of disorder of a dynamical system.

\begin{definition}
A function $f$ is said to have the property of \emph{expansivity} if
$$
\exists \varepsilon >0,\forall x\neq y,\exists n\in \mathbb{N}%
,d(f^{n}(x),f^{n}(y))\geqslant \varepsilon .
$$
\end{definition}

Then $\varepsilon $ is the \emph{constant of expansivity} of $f$: an
arbitrarily small error on any initial condition is \emph{always}
amplified until $\varepsilon $.

\begin{proposition}
The folding process in the 2D model is an expansive chaotic system on
$(\mathcal{X},d)$. Its constant of expansivity is at least equal to 1.
\end{proposition}

\begin{proof}
Let $X=(C,F)$ and $X'=(C',F')$ such that $X \neq X'$.
\begin{itemize}
\item If $C \neq C'$, then $\exists k_0 \in \llbracket 1;\mathsf{N}
\rrbracket, C_{k_0} \neq C_{k_0}'$. So, $$d\left( G^0(X) , G^0(X')
\right) \geqslant 2^{\mathsf{N}-k_0} \geqslant 1.$$
\item Else $F' \neq F$. Let $k_0 = min \left\{ k \in \mathds{N}, F^k
\neq F'^k \right\}.$ Then $\forall k < k_0, G^k(X) = G^k(X')$. Let
$\check{X} = (\check{C}, \check{F}) = G^{k_0-1}(X) =
G^{k_0-1}(X')$.

\begin{flushleft}
Then $d\left( G^{k_0}\left(X\right), G^{k_0}\left(X'\right)\right)$
\end{flushleft}
\vspace{-0.5cm}
\begin{flushright}
$\begin{array}{l} \geqslant d_C\left( f_{F^{k_0}}\left(\check{C}_1,
\hdots, \check{C}_\mathsf{N}\right), f_{F'^{k_0}}\left(\check{C}_1,
\hdots, \check{C}_\mathsf{N}\right)\right)\\ \geqslant d_C\left(
\left(\check{C}_1, \hdots,\check{C}_{\left|F^{k_0} \right|-1},
f^{sign(F^{k_0})}\left(\check{C}_{\left|F^{k_0}\right|}\right),
\hdots,\right.\right. \\
\left.f^{sign(F^{k_0})}\left(\check{C}_{\mathsf{N}}\right)\right),
\left(\check{C}_1, \hdots,\check{C}_{\left|F'^{k_0}
\right|-1},\right.\\
\left.\left.f^{sign(F'^{k_0})}\left(\check{C}_{\left|F'^{k_0}\right|}\right),
\hdots,
f^{sign(F'^{k_0})}\left(\check{C}_{\mathsf{N}}\right)\right)\right)\\
\geqslant 2^{\mathsf{N}-max\left(\left|F^{k_0}\right|,
\left|F'^{k_0}\right|\right)} \\ \geqslant 1.
\end{array}$
\end{flushright}
\end{itemize}
\end{proof}
So the result is established.

\section{Consequences}
\label{Sec:Consequences}

\subsection{Is Chaotic Behavior Incompatible with Approximately one 
Thousand Folds?}

Results established previously only concern the folding process in the
2D HP square lattice model.  At this point, it is natural to wonder if
such a  model, being  a reasonable approximation  of the  true natural
process,  is chaotic  because  this natural  process  is chaotic  too.
Indeed, claiming  that the natural protein folding  process is chaotic
seems to  be contradictory with  the fact that only  approximately one
thousand    folds   have    been   discovered    this    last   decade
\cite{Andreeva01012004}.   The   number  of  proteins   that  have  an
understood 3D  structure increases  largely year after  year.  However
the number of  new categories of folds seems to be  limited by a fixed
value  approximately  equal to  one  thousand.   Indeed,  there is  no
contradiction as a chaotic behavior  does not forbid a certain form of
order.   As  stated  before,  chaos  only  refers  to  limitations  in
prediction.  For  example, seasons are  not forbidden even  if weather
forecast  has a  non-intense chaotic  behavior.  A  similar regularity
appears in brains: even if hazard  and chaos play an important role on
a  microscopic  scale,  a  statistical  order appears  in  the  neural
network.
 
That is, a  certain order can emerge from a  chaotic behavior, even if
it  is not  a rule  of  thumb. More  precisely, in  our opinion  these
thousand  folds can  be related  to basins  of attractions  or strange
attractors of the dynamical system, objects that are well described by
the  mathematical theory  of chaos.  Thus,  it should  be possible  to
determine all of  the folds that can occur, by  refining our model and
looking  for  its  basins   of  attractions  with  topological  tools.
However, this assumption still remains to be investigated.

\subsection{Is Artificial Intelligence able to Predict Chaotic Dynamic?}

We  will  now  focus on  the  impact  of  using  a chaotic  model  for
prediction.  We  give some results  on two kinds of  experiments, both
using   neural   networks.   Firstly,   we   will   study  whether   a
(mathematical) chaotic behavior can be  learned by a neural network or
not.  Therefore, we design a  global recurrent network that models the
function $F_f$ introduced in the  previous section and we show that it
is more  difficult to  train the network  when $f$ is  chaotic.  These
considerations have been formerly  proposed in \cite{bgs11:ip} and are
extended here.  Secondly, we will try to learn the future conformation
of proteins that consist of  a small number of residues. Our objective
is to  assess if  a neural network  can learn the  future conformation
given the current one and a sequence of a few folds.

In  this  work,  we  choose   to  train  a  classical  neural  network
architecture:  the MultiLayer  Perceptron, a  model of  network widely
used   and  well-known  for   its  universal   approximation  property
\cite{DBLP:journals/nn/HornikSW89}. Let  us notice that  for the first
kind of experiments global feedback connections are added, in order to
have a proper  modeling of chaotical iterations, while  for the latter
kind of experiments the MLPs used are feed-forward ones. In both cases
we  consider  networks  having  sigmoidal hidden  neurons  and  output
neurons with a linear activation  function. They are trained using the
Limited-memory Broyden-Fletcher-Glodfarb-Shanno quasi-Newton algorithm
with  Wolfe  linear  search.   The  training  process  can  either  be
controlled by  the number of  network parameters (weights  and biases)
updates, also called epochs, or by a mean square error criterion.

\subsubsection{Can a Neural Network Learn Chaotic Functions?}

\smallskip

{\it Experimental Protocol}

\smallskip

We consider $f:\mathds{B}^\mathsf{N} \longrightarrow
\mathds{N}^\mathsf{N}$, strategies of singletons ($\forall n \in
\mathds{N}, S^n \in \llbracket 1; \mathsf{N} \rrbracket$), and a MLP
that recognize $F_{f}$. That means, for all $(k,x) \in \llbracket 1 ;
\mathsf{N} \rrbracket \times \mathds{B}^\mathsf{N}$, the response of
the output layer to the input $(k,x)$ is $F_{f}(k,x)$. We thus
connect the output layer to the input one as it is depicted in
Fig.~\ref{perceptron}, leading to a global recurrent artificial neural
network working as follows \cite{bgs11:ip}.

\begin{figure}[b]
\centering
\includegraphics[width=3.25in]{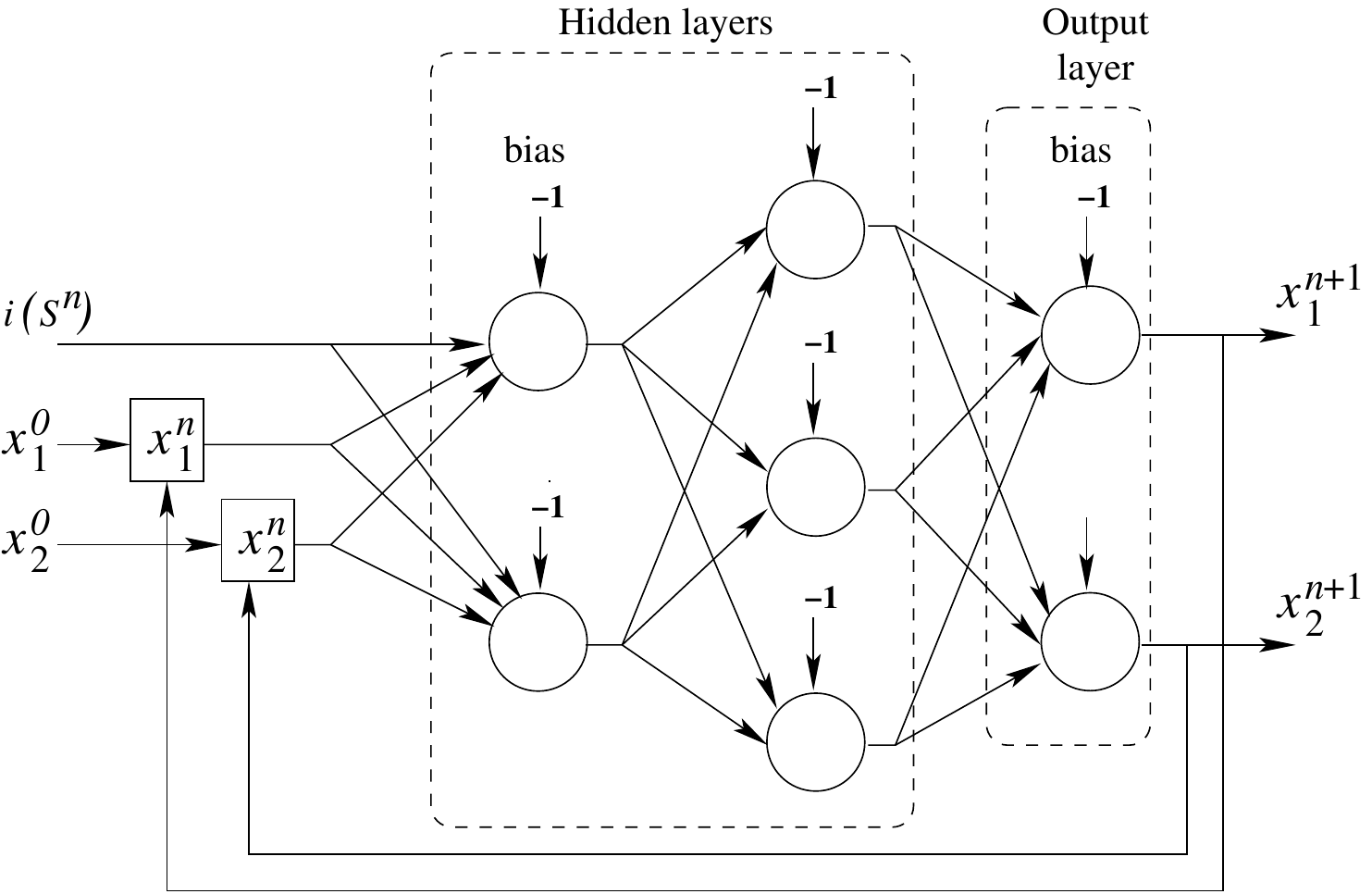}
\caption{Recurrent neural network modeling $F_{f}$}
\label{perceptron}
\hfil
\end{figure}

At  the initialization stage,  the network  receives a  Boolean vector
$x^0\in\mathds{B}^\mathsf{N}$ as input  state, and $S^0 \in \llbracket
1;\mathsf{N}\rrbracket$ in its input integer channel $i()$. Thus, $x^1
= F_{f}(S^0,  x^0)\in\mathds{B}^\mathsf{N}$ is computed  by the neural
network.  This state $x^1$  is published  as an  output. Additionally,
$x^1$ is sent back to the input  layer, to act as Boolean state in the
next iteration. Finally, at iteration number $n$, the recurrent neural
network  receives  the  state $x^n\in\mathds{B}^\mathsf{N}$  from  its
output     layer     and     $i\left(S^n\right)     \in     \llbracket
1;\mathsf{N}\rrbracket$ from  its input integer channel  $i()$. It can
thus      calculate      $x^{n+1}      =      F_{f}(i\left(S^n\right),
x^n)\in\mathds{B}^\mathsf{N}$,  which will  be the  new output  of the
network.  Obviously, this  particular  MLP produces  exactly the  same
values  as  CIs with  update  function $f$.  That  is,  such MLPs  are
equivalent,  when working  with  $i(s)$,  to CIs  with  $f$ as  update
function and strategy $S$ \cite{bgs11:ip}.

Let us now introduce the two following functions:
\begin{itemize}
\item $f_1(x_{1},x_2,x_3)=(\overline{x_{1}},\overline{x_{2}},\overline{x_{3}})$,
\item $f_2(x_{1},x_2,x_3)=(\overline{x_{1}},x_{1},x_{2})$.
\end{itemize}
It can easily  be checked that these functions  satisfy the hypothesis
of Theorem \ref{Th:Caracterisation des  IC chaotiques}, thus their CIs
are  chaotic  according  to  Devaney.   Then,  when  the  MLP  defined
previously  learns to recognize  $F_{f_1}$ or  $F_{f_2}$, it  tries to
learn these  CIs, that  is, a chaotic  behavior as defined  by Devaney
\cite{bgs11:ip}. On the contrary, the function
$$
g(x_{1},x_2,x_3)=(\overline{x_{1}},x_{2},x_{3})
$$ is such  that $\Gamma(g)$ is not strongly  connected. In this case,
due  to Theorem  \ref{Th:Caracterisation des  IC chaotiques},  the MLP
does not learn a chaotic process.  We will now recall the study of the
training  process  of  functions  $F_{f_1}$,  $F_{f_2}$,  and  $F_{g}$
\cite{bgs11:ip}, that is to say, the ability to learn one iteration of
CIs.

\medskip

\noindent {\it Experimental Results}

\smallskip

For each neural network we  have considered MLP architectures with one
and two  hidden layers,  with in the  first case different  numbers of
hidden  neurons. Thus  we will  have  different versions  of a  neural
network modeling the same iteration function \cite{bgs11:ip}. Only the
size  and number of  hidden layers  may change,  since the  numbers of
inputs and  output neurons  are fully specified  by the  function. The
training is performed until the  learning error is lower than a chosen
threshold value ($10^{-2}$).

\begin{table}[!t]
\renewcommand{\arraystretch}{1.3}
\caption{Results of some iteration functions learning, using different
recurrent MLP architectures}
\label{tab1}
\centering
\begin{scriptsize}
\begin{tabular}{|c||c|c|c|c|}
\hline 
 & \multicolumn{4}{c|}{One hidden layer} \\
\cline{2-5}
 & \multicolumn{2}{c|}{8 neurons} & \multicolumn{2}{c|}{10 neurons} \\
\hline
Function & Mean & Success & Mean & Success \\
 & epoch & rate & epoch & rate \\
\hline
$f_1$ & 82.21 & 100\% & 73.44 & 100\% \\
$f_2$ & 76.88 & 100\% & 59.84 & 100\% \\
$g$ & 36.24 & 100\% & 37.04 & 100\% \\
\hline
\hline
 & \multicolumn{4}{c|}{Two hidden layers: 8 and 4 neurons} \\
\cline{2-5}
 & \multicolumn{2}{c|}{Mean epoch number} & \multicolumn{2}{|c|}{Success rate} \\
\hline
$f_1$ & \multicolumn{2}{c|}{203.68} & \multicolumn{2}{c|}{76\%} \\
$f_2$ & \multicolumn{2}{c|}{135.54} & \multicolumn{2}{c|}{96\%} \\
$g$ & \multicolumn{2}{c|}{72.56} & \multicolumn{2}{c|}{100\%} \\
\hline
\end{tabular}
\end{scriptsize}
\end{table}

Table~\ref{tab1}  gives for  each considered  neural network  the mean
number of  epochs needed to  learn one iteration  in their ICs,  and a
success rate  that reflects  a successful training  in less  than 1000
epochs. Both values are  computed considering 25 trainings with random
weights  and biases  initialization. These  results  highlight several
points \cite{bgs11:ip}. First, the two hidden layer structure seems to
be  quite  inadequate to  learn  chaotic  behaviors. Second,  training
networks so  that they  behave chaotically seems  to be  difficult for
these simplistic functions only iterated  one time, since they need in
average more  epochs to be correctly  trained. In the case  of the two
hidden  layer network  topology, a  comparison of  the mean  number of
epochs needed  for a successful learning of  10~chaotic functions with
that  obtained for  10~non  chaotic functions reinforces the  previous
observation.  Indeed,  the  learning  of chaotic  functions  needs  in
average   284.57~epochs,   whereas   non~chaotic   functions   require
232.87~epochs.  In the future, we  also plan to consider larger values
for~$\mathsf{N}$.

\subsubsection{Can a Neural Network Predict a Future Protein Conformation?}

\smallskip

{\it Experimental Protocol}

\smallskip

In this second set of  experiments, multilayer perceptrons are used to
learn the conformation of  very simple proteins (peptides, indeed). In
fact, we consider proteins composed  of five residues, of which only 4
can change since the first one  is always $0$, and folding dynamics of
two or three  folds. For example, if the  current protein conformation
is $(0)1222$,  and folds $4$  and $-1$ are successively  applied, then
the new conformation will be $(0)0112$. Obviously, these choices, that
lead respectively to 20736  and 186624 potential conformations, do not
correspond  to realistic  folding  processes. However,  they allow  to
evaluate  the  ability  of   neural  networks  to  learn  very  simple
conformations.

The  networks  consist  of  MLP   with  3  or  4~inputs,  the  current
conformation  without  the first  residue,  and  a  sequence of  2  or
3~successive  folds.  It  produces  a  single  output:  the  resulting
conformation.  Additionally,  we  slightly  change the  classical  MLP
structure in order to improve  the capacity of such neural networks to
model nonlinear relationships and  to be trained faster. Therefore, we
retain    the   HPU    (Higher-order   Processing    Unit)   structure
\cite{DBLP:journals/ijns/GhoshS92}. This latter artificially increases
the number of inputs by  adding polynomial combinations of the initial
inputs  up to  a given  degree, called  the order  of the  network. To
prevent overfitting  and to  assess the generalization  performance we
use holdout  validation, which means that  the data set  is split into
learning,  validation, and  test subsets.  These subsets  are obtained
through a random sampling strategy.

To  estimate  the  prediction  accuracy  we  use  the  coefficient  of
variation of the root  mean square error (CV\-RMSE), usually presented
as  a  percentage,  the  average  relative  variance  (AVR),  and  the
coefficient of efficiency  denoted by (E). These measures  give a good
estimation of  the capacity of a  neural network to  explain the total
variance of the data. The CVRMSE of the prediction is defined as:
$$
\mbox{CVRMSE}=\frac{100}{\overline{e_k}} \cdot 
 \sqrt{\frac{\sum_{k=1}^N \left(e_k-p_k\right)^2}{N}},
$$ 
where $e_k$ is the expected output for the $k-$th input-output pair,
$p_k$ is the predicted output, $N$ is the number of pairs in the test
set, and $\overline{e_k}$ is the mean value of the expected output.
The average relative variance and coefficient of efficiency are
respectively expressed by:
$$
\mbox{ARV}=\frac{\sum_{k=1}^N \left(e_k-p_k\right)^2}
 {\sum_{k=1}^N \left(e_k-\overline{e_k}\right)^2} \mbox{ and }
 \mbox{E}=1-\mbox{ARV}.
$$ 
These values reflect the accuracy of the prediction as follows: the
nearer CVRMSE and AVR to $0$, and consequently $E$ close to 1.0, the
better the prediction.

\medskip

\noindent {\it Experimental Results}

\smallskip

The considered neural  networks differ in the size  of a single hidden
layer and are trained until a  maximum number of epochs is reached. As
we  set the  order  of the  HPU  structure to~$3$,  instead  of 3  and
4~initial inputs we have 19 and 34~inputs. We train neural networks of
15 and 25~hidden neurons, using as maximum number of epochs a value in
$\{500,1000,2500\}$,  whatever  the  number  of  initial  inputs  (see
Table~\ref{tab2}).  The  learning, validation,  and  test subsets  are
built in such  a way that they respectively  represent 65\%, 10\%, and
25\% of the whole data set.  In the case of 3~initial inputs data sets
of 5184, 10368, and 15552~samples are used, they represent 25\%, 50\%,
and  75\%  of the  20736~potential  conformations.  For the  4~initial
outputs  case   we  restrict  our   experiments  to  a  data   set  of
46656~samples,  that  corresponds  to  25\%  of  the  186624~potential
conformations.

In  Table~\ref{tab2} we give,  for the  different learning  setups and
data  set  sizes,  the mean  values  of  CVMSE,  AVR,  and E  for  the
output.  To compute  these  values, 10~trainings  with random  subsets
construction and network  parameters initialization were performed. It
can be seen that in all cases the better performances are obtained for
the  larger networks  (25~hidden neurons)  that are  the  most trained
(2500~epochs).  Furthermore, a larger  data set  allows only  a slight
increase  of the  prediction  quality.  We observe  for  a three  time
increase of  the data  set: from 5184  to 15552~samples, a  very small
improvement of 1.66\% for the coefficient~$E$.

At a first  glance, the prediction accuracy seems not  too bad for the
3~initial  inputs  topology,  with  coefficients of  efficiency  above
0.9.   However,  remember  that   we  try   to  predict   very  simple
conformations far away from the realistic ones. Furthermore, a look at
the results obtained  for the second topology, the  one with~4~initial
outputs, shows that predicting  a conformation that undergoes only one
more folding  transition is  intractable with the  considered learning
setups: the efficiency coefficient  is always below 0.5.  Clearly, the
different  neural  networks  have  failed at  predicting  the  protein
folding  dynamics. A  larger  neural network  with  a longer  training
process may be  able to improve these results.  But finding a learning
setup well  suited for the prediction of  relevant proteins structures
that are far more complex seems very hypothetical.

\begin{table}[!b]
\renewcommand{\arraystretch}{1.3}
\caption{Results of the validation of networks with an HPU structure of
order~3 for several numbers of hidden neurons}
\label{tab2}
\centering
\begin{scriptsize}
\begin{tabular}{|c||c|c|c|c|}
\hline
Topology & \multicolumn{4}{c|}{3 initial~/~19 HPU inputs and 1 output} \\
\cline{2-5}
Hidden neurons & Epochs & \%~CVRMSE & ARV & E \\
\hline
 & \multicolumn{4}{|c|}{Data set of 5184 samples} \\
\cline{2-5}
\multirow{3}{*}{15 neurons} & 500 & 24.97 & 0.3824 & 0.6176 \\
 & 1000 & 20.67 & 0.2628 & 0.7372 \\
 & 2500 & 16.69 & 0.1731 & 0.8269 \\
\cline{2-5}
\multirow{3}{*}{25 neurons} & 500 & 23.33 & 0.3373 & 0.6627 \\
 & 1000 & 15.94 & 0.1565 & 0.8435 \\
 & 2500 & 10.75 & 0.0715 & 0.9285 \\
\hline
 & \multicolumn{4}{|c|}{Data set of 10368 samples} \\
\cline{2-5}
\multirow{3}{*}{15 neurons} & 500 & 26.27 & 0.4223 & 0.5777 \\
 & 1000 & 22.08 & 0.3000 & 0.7000 \\
 & 2500 & 18.81 & 0.2225 & 0.7775 \\
\cline{2-5}
\multirow{3}{*}{25 neurons} & 500 & 24.54 & 0.3685 & 0.6315 \\
 & 1000 & 16.11 & 0.1591 & 0.8409 \\
 & 2500 & 9.43 & 0.0560 & 0.9440 \\
\hline
 & \multicolumn{4}{|c|}{Data set of 15552 samples} \\
\cline{2-5}
\multirow{3}{*}{15 neurons} & 500 & 24.74 & 0.3751 & 0.6249 \\
 & 1000 & 19.92 & 0.2444 & 0.7556 \\
 & 2500 & 16.35 & 0.1659 & 0.8341 \\
\cline{2-5}
\multirow{3}{*}{25 neurons} & 500 & 22.90 & 0.3247 & 0.6753 \\
 & 1000 & 15.42 & 0.1467 & 0.8533 \\
 & 2500 & 8.89 & 0.0501 & 0.9499 \\
\hline
\hline
Topology & \multicolumn{4}{c|}{4 initial~/~34 HPU inputs and 1 output} \\
\hline
 & \multicolumn{4}{|c|}{Data set of 46656 samples} \\
\cline{2-5}
\multirow{3}{*}{15 neurons} & 500 & 35.27 & 0.7606 & 0.2394 \\
 & 1000 & 33.50 & 0.6864 & 0.3136 \\
 & 2500 & 31.94 & 0.6259 & 0.3741 \\
\cline{2-5}
\multirow{3}{*}{25 neurons} & 500 & 35.05 & 0.7535 & 0.2465 \\
 & 1000 & 32.25 & 0.6385 & 0.3615 \\
 & 2500 & 28.61 & 0.5044 & 0.4956 \\
\hline
\end{tabular}
\end{scriptsize}
\end{table}

Finally, let  us notice that the  HPU structure has a  major impact on
the  learning quality.  Indeed,  let us  consider  the coefficient  of
efficiency obtained  for the  data set of  5184~samples and  a network
composed of 25~hidden neurons trained  during 2500 epochs. As shown in
Table~2 the coefficient $E$ is  about 0.9285 if the neural network has
an  HPU  structure  of   order~3,  whereas  experiments  made  without
increasing the number of initial  inputs give 0.6930 as mean value for
$E$.  Similar experiments  in case  of the  second topology  result in
$E=0.2154$  for the  classical structure  that has  no  HPU structure.
That represents a  respective decrease of more than  25 and 50\%, from
which we  can say that MLP  networks with a  classical structure would
have given worse predictions.

At this point we can only claim that it is not completely evident that
computational  intelligence tools  like  neural networks  are able  to
predict,  with a  good accuracy,  protein folding.  To  reinforce this
belief, tools optimized to chaotic  behaviors must be found -- if such
tools exist.  Similarly, there should  be a link between  the training
difficulty and  the ``quality'' of  the disorder induced by  a chaotic
iteration  function  (their  constants  of  sensitivity,  expansivity,
etc.), and this second relation must be found.

\section{Conclusion}
\label{Conclusion}

In this  paper the  topological dynamics of  protein folding  has been
evaluated.  More  precisely,  we  have studied  whether  this  folding
process  is predictable  in the  2D model  or not.  It is  achieved to
determine if it is reasonable to think that computational intelligence
tools like  neural networks  are able  to predict the  3D shape  of an
amino  acids sequence.   It  is mathematically  proven,  by using  two
different  ways, that  protein folding  in  2D hydrophobic-hydrophilic
(HP) square lattice model is chaotic according to Devaney.

Consequences  for  both  structure  prediction and  biology  are  then
outlined.  In particular,  the  first comparison  of  the learning  by
neural networks of  a chaotic behavior on the one hand,  and of a more
natural dynamics on the other  hand, are outlined. The results tend to
show  that such  chaotic behaviors  are more  difficult to  learn than
non-chaotic  ones.  It is  not  our pretension  to  claim  that it  is
impossible to  predict chaotic behaviors such as  protein folding with
computational intelligence.  Our opinion  is just that  this important
point must now be regarded with attention.

In future work  the dynamical behavior of the  protein folding process
will be more deeply studied, by using topological tools as topological
mixing, Knudsen  and Li-Yorke  notions of chaos,  topological entropy,
etc. The quality  and intensity of this chaotic  behavior will then be
evaluated. Consequences both on folding prediction and on biology will
then be regarded in detail. This study may also allow us to determine,
at  least to  a certain  extent, what  kind of  errors on  the initial
condition lead to acceptable results, depending on the intended number
of  iterations (i.e.,  the number  of  folds). Such  a dependence  may
permit to define strategies depending on  the type and the size of the
proteins, their proportion of hydrophobic residues, and so on.

Other  molecular or  genetic  dynamics will  be  investigate by  using
mathematical topology, and other  chaotic behaviors will be looked for
(as  neurons in the  brain).  More  specifically, various  tools taken
from  the  field of  computational  intelligence  will  be studied  to
determine if some of these tools are capable to predict behaviors that
are  chaotic  with  a  good  accuracy.  It  is  highly  possible  that
prediction  depends  both  on  the  tool and  on  the  chaos  quality.
Moreover, the study  presented in this paper will  be extended to high
resolution 3D models.  Impacts of  the chaotic behavior of the protein
folding  process in  biology  will be  regarded.   Finally, the  links
between  this established  chaotic behavior  and stochastic  models in
gene expression, mutation, or in Evolution, will be investigated.


\bibliographystyle{spphys}    

\begin{thebibliography}{}
\providecommand{\url}[1]{{#1}}
\providecommand{\urlprefix}{URL }
\expandafter\ifx\csname urlstyle\endcsname\relax
  \providecommand{\doi}[1]{DOI \discretionary{}{}{}#1}\else
  \providecommand{\doi}{DOI \discretionary{}{}{}\begingroup
  \urlstyle{rm}\Url}\fi

\end{thebibliography}


\begin{thebibliography}{10}
\providecommand{\url}[1]{{#1}}
\providecommand{\urlprefix}{URL }
\expandafter\ifx\csname urlstyle\endcsname\relax
  \providecommand{\doi}[1]{DOI \discretionary{}{}{}#1}\else
  \providecommand{\doi}{DOI \discretionary{}{}{}\begingroup
  \urlstyle{rm}\Url}\fi
\bibitem{Hoque09}
M.~Hoque, M.~Chetty, A.~Sattar, in \emph{Biomedical Data and Applications},
  \emph{Studies in Computational Intelligence}, vol. 224, ed. by A.~Sidhu,
  T.~Dillon (Springer Berlin Heidelberg, 2009), pp. 317--342

\bibitem{Crescenzi98}
P.~Crescenzi, D.~Goldman, C.~Papadimitriou, A.~Piccolboni, M.~Yannakakis, in
  \emph{Proceedings of the thirtieth annual ACM symposium on Theory of
  computing} (ACM, New York, NY, USA, 1998), STOC '98, pp. 597--603

\bibitem{DBLP:conf/cec/HiggsSHS10}
T.~Higgs, B.~Stantic, T.~Hoque, A.~Sattar, in \emph{IEEE Congress on
  Evolutionary Computation} \cite{DBLP:conf/cec/2010}, pp. 1--8

\bibitem{Shmygelska2005Feb}
A.~Shmygelska, H.H. Hoos.
\newblock An ant colony optimisation algorithm for the 2d and 3d hydrophobic
  polar protein folding problem (2005 Feb)

\bibitem{DBLP:conf/cec/Perez-HernandezRG10}
L.G. P{\'e}rez-Hern{\'a}ndez, K.~Rodr\'{\i}guez-V{\'a}zquez,
  R.~Gardu{\~n}o-Ju{\'a}rez, in \emph{IEEE Congress on Evolutionary
  Computation} \cite{DBLP:conf/cec/2010}, pp. 1--8

\bibitem{Islam:2009:NMA:1695134.1695181}
M.K. Islam, M.~Chetty, in \emph{Proceedings of the 22nd Australasian Joint
  Conference on Advances in Artificial Intelligence} (Springer-Verlag, Berlin,
  Heidelberg, 2009), AI '09, pp. 412--421

\bibitem{Dubchak1995}
I.~Dubchak, I.~Muchnik, S.R. Holbrook, S.H. Kim, Proc Natl Acad Sci U S A
  \textbf{92}(19), 8700 (1995)

\bibitem{Anfinsen20071973}
C.B. Anfinsen, Science \textbf{181}(4096), 223 (1973).
\newblock \doi{10.1126/science.181.4096.223}.
\newblock \urlprefix\url{http://www.sciencemag.org/content/181/4096/223.short}

\bibitem{Bonneau01}
R.~Bonneau, D.~Baker, Annual Review of Biophysics and Biomolecular Structure
  \textbf{30}(1), 173 (2001).
\newblock \doi{10.1146/annurev.biophys.30.1.173}

\bibitem{Chivian2005}
D.~Chivian, D.E. Kim, L.~Malmström, J.~Schonbrun, C.A. Rohl, D.~Baker,
  Proteins \textbf{61}(S7), 157 (2005).
\newblock \urlprefix\url{http://dx.doi.org/10.1002/prot.20733}

\bibitem{Zhang2005}
Y.~Zhang, A.K. Arakaki, J.~Skolnick, Proteins \textbf{61}(S7), 91 (2005).
\newblock \urlprefix\url{http://dx.doi.org/10.1002/prot.20724}

\bibitem{bgc11:ip}
J.~Bahi, C.~Guyeux, N.~Cote, in \emph{IJCNN 2011, Int. Joint Conf. on Neural
  Networks} (San Jose, California, United States, 2011), pp. ***--***.
\newblock To appear

\bibitem{Bohm1991375}
G.~B{\"o}hm, Chaos, Solitons \& Fractals \textbf{1}(4), 375  (1991).
\newblock \doi{DOI: 10.1016/0960-0779(91)90028-8}.
\newblock
  \urlprefix\url{http://www.sciencedirect.com/science/article/B6TJ4-46CBXVT-1X%
/2/370489c218e4c2732cd9b620ef50c696}

\bibitem{Zhou96}
H.b. Zhou, L.~Wang, The Journal of Physical Chemistry \textbf{100}(20), 8101
  (1996).
\newblock \doi{10.1021/jp953409x}

\bibitem{Braxenthaler97}
M.~Braxenthaler, R.R. Unger, D.~Auerbach, J.~Moult, Proteins-structure Function
  and Bioinformatics \textbf{29}, 417 (1997).
\newblock \doi{10.1002/(SICI)1097-0134(199712)29:4<417::AID-PROT2>3.3.CO;2-O}

\bibitem{Berger98}
B.~Berger, T.~Leighton, in \emph{Proceedings of the second annual international
  conference on Computational molecular biology} (ACM, New York, NY, USA,
  1998), RECOMB '98, pp. 30--39

\bibitem{Dill1985}
K.~Dill, Biochemistry \textbf{24}(6), 1501 (1985).
\newblock \urlprefix\url{http://ukpmc.ac.uk/abstract/MED/3986190}

\bibitem{DBLP:conf/cec/IslamC10}
M.K. Islam, M.~Chetty, in \emph{IEEE Congress on Evolutionary Computation}
  \cite{DBLP:conf/cec/2010}, pp. 1--8

\bibitem{Unger93}
R.~Unger, J.~Moult, in \emph{Proceedings of the 5th International Conference on
  Genetic Algorithms} (Morgan Kaufmann Publishers Inc., San Francisco, CA, USA,
  1993), pp. 581--588

\bibitem{DBLP:conf/cec/HorvathC10}
D.~Horvath, C.~Chira, in \emph{IEEE Congress on Evolutionary Computation}
  \cite{DBLP:conf/cec/2010}, pp. 1--8

\bibitem{Dyson2005}
H.J. Dyson, P.E. Wright, Nature Reviews Molecular Cell Biology \textbf{6}(3),
  197 (2005).
\newblock \doi{DOI: 10.1038/nrm1589}

\bibitem{doi:10.1146/annurev.biophys.37.032807.125924}
V.N. Uversky, C.J. Oldfield, A.K. Dunker, Annual Review of Biophysics
  \textbf{37}(1), 215 (2008).
\newblock \doi{10.1146/annurev.biophys.37.032807.125924}.
\newblock PMID: 18573080

\bibitem{guyeux09}
J.M. Bahi, C.~Guyeux, Journal of Algorithms \& Computational Technology
  \textbf{4}(2), 167 (2010)

\bibitem{Shmygelska05}
A.~Shmygelska, H.~Hoos, BMC Bioinformatics \textbf{6}(1), 30 (2005).
\newblock \doi{10.1186/1471-2105-6-30}

\bibitem{Backofen99algorithmicapproach}
R.~Backofen, S.~Will, P.~Clote.
\newblock Algorithmic approach to quantifying the hydrophobic force
  contribution in protein folding (1999)

\bibitem{Devaney}
R.L. Devaney, \emph{An Introduction to Chaotic Dynamical Systems}, 2nd edn.
  (Addison-Wesley, Redwood City, CA, 1989)

\bibitem{Banks92}
J.~Banks, J.~Brooks, G.~Cairns, P.~Stacey, Amer. Math. Monthly \textbf{99}, 332
  (1992)

\bibitem{guyeuxTaiwan10}
J.M. Bahi, C.~Guyeux, Q.~Wang, in \emph{ICCASM 2010, Int. Conf. on Computer
  Application and System Modeling} (Taiyuan, China, 2010), pp.
  V13--643--V13--647.
\newblock \doi{10.1109/ICCASM.2010.5622199}.
\newblock \urlprefix\url{http://dx.doi.org/10.1109/ICCASM.2010.5622199}

\bibitem{guyeux10}
J.M. Bahi, C.~Guyeux, in \emph{WCCI'10, IEEE World Congress on Computational
  Intelligence} (Barcelona, Spain, 2010), pp. 1--7.
\newblock Best paper award

\bibitem{bcgr11:ip}
J.~Bahi, J.f. Couchot, C.~Guyeux, A.~Richard, in \emph{FCT'11, 18th Int. Symp.
  on Fundamentals of Computation Theory}, \emph{LNCS}, vol. 6914 (Oslo, Norway,
  2011), \emph{LNCS}, vol. 6914, pp. 126--137

\bibitem{Andreeva01012004}
A.~Andreeva, D.~Howorth, S.E. Brenner, T.J.P. Hubbard, C.~Chothia, A.G. Murzin,
  Nucleic Acids Research \textbf{32}(suppl 1), D226 (2004).
\newblock \doi{10.1093/nar/gkh039}

\bibitem{bgs11:ip}
J.~Bahi, C.~Guyeux, M.~Salomon, in \emph{ICCANS 2011, IEEE Int. Conf. on
  Computer Applications and Network Security} (Maldives, Maldives, 2011)

\bibitem{DBLP:journals/nn/HornikSW89}
K.~Hornik, M.B. Stinchcombe, H.~White, Neural Networks \textbf{2}(5), 359
  (1989)

\bibitem{DBLP:journals/ijns/GhoshS92}
J.~Ghosh, Y.~Shin, Int. J. Neural Syst. \textbf{3}(4), 323 (1992)

\bibitem{DBLP:conf/cec/2010}
\emph{Proceedings of the IEEE Congress on Evolutionary Computation, CEC 2010,
  Barcelona, Spain, 18-23 July 2010} (IEEE, 2010)

\end{thebibliography}

%
%

\end{document}